\documentclass[11pt]{article}
\usepackage{amsmath,amsfonts,amssymb,amsthm}
\usepackage{fullpage}
\usepackage[colorlinks]{hyperref}
\hypersetup{linkcolor=cyan,filecolor=cyan,citecolor=cyan,urlcolor=cyan}
\usepackage{xspace}
\usepackage{thm-restate,color,xcolor}
\usepackage{boxedminipage}
\usepackage[boxed]{algorithm}
\usepackage{epigraph}
\usepackage[sc]{mathpazo}
\usepackage{amsmath}
\usepackage{mathtools}
\usepackage{framed}
\usepackage[framemethod=tikz]{mdframed}
\usepackage{titlesec}
\usepackage{lipsum}%
\usepackage{cleveref,aliascnt}
\usepackage{tikz}
\usepackage{wrapfig}
\usepackage{framed}
\usepackage[framemethod=tikz]{mdframed}
\usepackage{multirow}
\usepackage{multicol}
\newcommand{\specialcellleft}[2][l]{%
  \begin{tabular}[#1]{@{}l@{}}#2\end{tabular}}
\newcommand{\specialcellcenter}[2][c]{%
  \begin{tabular}[#1]{@{}c@{}}#2\end{tabular}}

\usetikzlibrary{shapes.geometric}
\usetikzlibrary{arrows}
\usetikzlibrary{arrows.meta}
\usetikzlibrary{patterns}
\usetikzlibrary{shapes.misc}

\newcommand{\dist}{\mbox{\rm dist}}
\newcommand{\len}{\mbox{\rm len}}

\newcommand{\HopsetToPartialSpanner}{\mathsf{HopsetToMissingSpanner}}

\newcommand{\eps}{\varepsilon}

\theoremstyle{plain}
\newtheorem{thm}{Theorem}[section]
\newcommand{\BTHM}{\begin{thm}} \newcommand{\ETHM}{\end{thm}}
\newtheorem{cor}[thm]{Corollary}
\newcommand{\BCR}{\begin{cor}} \newcommand{\ECR}{\end{cor}}
\newtheorem{lem}[thm]{Lemma}
\newcommand{\BL}{\begin{lem}}   \newcommand{\EL}{\end{lem}}
\newtheorem{clm}[thm]{Claim}
\newcommand{\BCM}{\begin{clm}}   \newcommand{\ECM}{\end{clm}}
\newtheorem{prop}[thm]{Proposition}
\newcommand{\BP}{\begin{prop}}   \newcommand{\EP}{\end{prop}}
\newtheorem{assm}[thm]{Assumption}
\newcommand{\BASM}{\begin{assm}}   \newcommand{\EASM}{\end{assm}}

\theoremstyle{definition}
\newtheorem{defn}{Definition}[section]
\newcommand{\BD}{\begin{defn}}   \newcommand{\ED}{\end{defn}}

\newtheorem{con}[thm]{Conjecture}
\newcommand{\BCONJ}{\begin{con}}   \newcommand{\ECONJ}{\end{con}}

\theoremstyle{inequality}
\newtheorem{inequality}{Inequality}[section]
\newtheorem{problem}[thm]{Problem}
\newcommand{\BPR}{\begin{problem}}   \newcommand{\EPR}{\end{problem}}

\newenvironment{rem}{\noindent{\bf Remark:~~}}{}
\newcommand{\BREM}{\begin{rem}} \newcommand{\EREM}{\end{rem}}
\newenvironment{discussion}{\noindent{\bf Discussion:~~\\}}{}
\newcommand{\BDIS}{\begin{discussion}} \newcommand{\EDIS}{\end{discussion}}
\newtheorem{obs}[thm]{Observation}
\numberwithin{equation}{section}

\def\blackslug
     {\hbox{\hskip 1pt\vrule width 8pt height 8pt depth 1.5pt\hskip 1pt}}
\def\qed{\quad\blackslug\lower 8.5pt\null\par}

\newcommand{\poly}{{\rm poly}}

\newcommand{\cA}{{\cal A}}

\newcommand{\APSP}{\mathsf{APSP}}

\newtheorem{exmp}[thm]{Example}
\newtheorem{fact}[thm]{Fact}
\newcommand{\BEX}{\begin{exmp}} \newcommand{\EEX}{\end{exmp}}
\newcommand{\BF}{\begin{fact}}   \newcommand{\EF}{\end{fact}}
\newcommand{\Bcr}{\begin{techcorr}}
\newcommand{\Ecr}{\end{techcorr}}
\newcommand{\BDS}{\begin{description}}
\newcommand{\EDS}{\end{description}}
\newcommand{\BE}{\begin{enumerate}}
\newcommand{\EE}{\end{enumerate}}
\newcommand{\BI}{\begin{itemize}}
\newcommand{\EI}{\end{itemize}}
\newcommand{\BPF}{\begin{proof}}
\newcommand{\EPF}{\end{proof}}

\title{Having Hope in Hops: \\New Spanners, Preservers and Lower Bounds for Hopsets \vspace{-7pt}}
\author{
Shimon Kogan  \vspace{-7pt}\\ 
        \small Weizmann Institute \vspace{-7pt}\\ 
        \small shimon.kogan@weizmann.ac.il
\and				
Merav Parter \thanks{This project is funded by the European Research Council (ERC) under the European Union’s Horizon 2020 research and innovation programme (grant agreement No. 949083).} \vspace{-7pt}\\ 
        \small Weizmann Institute \vspace{-7pt}\\ 
        \small merav.parter@weizmann.ac.il
}
\date{}

\begin{document}
\maketitle

\begin{abstract}
Hopsets and spanners are fundamental graph structures, playing a key role in shortest path computation, distributed communication, and more. A (near-exact) \emph{hopset} for a given graph $G$ is a (small) subset of weighted edges $H$ that when \emph{added} to the graph $G$ reduces the number of hops (edges) of near-exact shortest paths. Spanners and distance preservers, on the other hand, ask for \emph{removing} many edges from the graph while approximately preserving shortest path distances.

We provide a general reduction scheme from graph hopsets to the known metric compression schemes of spanners, emulators and distance preservers. Consequently, we get new and improved upper bound constructions for the latter, as well as, new lower bound results for hopsets. Our main results include:

\begin{itemize}
\item For $n$-vertex \textbf{directed weighted} graphs, one can provide $(1+\epsilon)$--approximate distance preservers\footnote{I.e., subgraphs that preserve the pairwise distances up to a multiplicative stretch of $(1+\epsilon)$.} for $p$ pairs in $V \times V$ with $\widetilde{O}_{\epsilon}(n \cdot p^{2/5}+(np)^{2/3})$ edges. For $p \geq n^{5/4}$, this matches the state-of-the art bounds for \emph{reachability} preservers by [Abboud and Bodwin, SODA 2018] and the lower bound for exact-distance preservers by [Bodwin, SODA 2016].

\item For $n$-vertex \textbf{undirected weighted} graphs, one can provide $(1+\epsilon)$ distance preserves with $\widetilde{O}_{\epsilon}(n^{1+o(1)}+p \cdot n^{o(1)})$ edges. So far, such bounds could be obtained only for \emph{unweighted} graphs. Consequently, we also get improved sourcewise spanners [Roditty, Thorup and Zwick, ICALP 2005] and spanners with slack [Chan, Dinitz and Gupta, ESA 2006].

\item Exact hopsets of linear size admit a worst-case hopbound of $\beta=\Omega(n^{1/3})$. This holds even for \textbf{undirected weighted} graphs, improving upon the $\Omega(n^{1/6})$ lower bound by [Huang and Pettie, SIAM J. Discret. Math 2021]. Interestingly this matches the recent diameter bound achieved for linear directed shortcuts.
\end{itemize}
More conceptually, our work makes a significant progress on the tantalizing open problem concerning the formal connection between hopsets and spanners, e.g., as posed by Elkin and Neiman [Bull. EATCS 2020]. 
 
\end{abstract}

\newpage

\newpage

\tableofcontents
\newpage
\setcounter{page}{1}

\section{Introduction}
\vspace{-5pt}
\subsection{Hopsets vs. Spanners, Emulators and Distance Preservers} \vspace{-5pt}

This paper is concerned with establishing a formal connection between graph hopsets and the well known graphical compression schemes of spanners, emulators and distance preservers. Our proposed reduction provides a new (and black-box) approach for computing distance preserving structures as well as lower bounds for hopsets. The collection of graph structures considered in this paper has been intensively studied over the last two decades \cite{PelegS89,AlthoferDDJ90,ElkinP04,BaswanaKMP05,Woodruff06,Pettie09,Chechik13,HuangP19,ElkinN19} due to their algorithmic centrality in the context of e.g., shortest path computation, routing, synchronization, distributed communication and beyond. While there has been evidence for algorithmic connections between these structures, no rigorous reductions have been known before. We start by providing a quick presentation of the studied graph structures.

\smallskip 

\textbf{Hopset.} Graph hopsets, introduced by Cohen \cite{Cohen00} (and informally by \cite{UllmanY91,KleinS97,ShiS99}) augment the input (possibly weighted and directed) graph $G=(V,E)$ with additional weighted edges in $V \times V$ to reduce the hopbound of approximate shortest paths. The primary objective in this context is in optimizing the tradeoff between the size of $H$, the stretch parameter $\epsilon$ and the hopbound $\beta$. The \emph{bounded-hop distance} $\dist^{(\beta)}_{G'}(u,v)$ stands for the length of the shortest path between $u$ and $v$ in $G'$ that has at most $\beta$ edges (hops). 
%
\vspace{-5pt}  
\begin{defn}[$(\beta,\eps)$-Hopsets]\label{def:hopset}
Given a (possibly weighted and directed) graph $G=(V,E)$, positive integer $\beta$ and $\epsilon \in (0,1)$, a subset $H \subseteq V \times V$ of weighted edges is called a $(\beta , \epsilon)$-hopset if for every $u, v \in V$, $\dist_G(u,v) \leq \dist^{(\beta)}_{G \cup H}(u, v) 
\leq (1+\epsilon)\dist_G(u,v)~$. 
The edges of the hopset $H$ are weighted by the length of the shortest path connecting their endpoints in $G$.
\end{defn}
\vspace{-5pt}  
Due to their importance in the context of shortest path computation, most notably in the parallel \cite{LiuJS19,ElkinN19,Fineman20,CaoFR20,ElkinM21}, distributed \cite{ElkinN16,ForsterN18} and dynamic \cite{HenzingerKN14} settings, hopsets have attracted a lot of activity, from combinatorial and algorithmic perspectives.  
The state-of-the-art bounds for \emph{undirected} $(\beta,\eps)$-hopsets are obtained by employing the Thorup-Zwick algorithm as shown independently by Huang and Pettie \cite{HuangP19} and Elkin and Neiman \cite{ElkinN19}. This algorithm has the remarkable property of providing universal $(\beta,\eps)$-hopsets with $O(n^{1+1/(2^{k+1}-1)})$ edges, for $\beta=O(k/\epsilon)^k$, for all values of $\epsilon$ \emph{simultaneously}. For constant values of $k$ and sufficiently small values of $\epsilon$ (as a function of $k$), Abboud, Bodwin and Pettie \cite{AbboudBP18} showed that this tradeoff is almost tight. Our understanding of \emph{directed} hopsets is considerably less complete. In one of the earliest papers on this topic, Ullman and Yannakakis \cite{UllmanY91} described a sampling-based approach to provide exact hopsets with hopbound $\beta$ and $\widetilde{O}((n/\beta)^2)$ edges. Until very recently, this was the only known upper bound, even if one settles for an arbitrarily large stretch (i.e., merely preserving reachability). A recent work of Kogan and Parter \cite{KoganParter22} provided an improved tradeoff, which for example, provides directed $(\beta,\eps)$-hopsets with $\beta=\widetilde{O}(n^{2/5})$ and linear number of edges, improving upon the state-of-the-art hopbound of $\sqrt{n}$. 
The current hopbound lower bound for linear hopsets is  $\Omega(n^{1/6})$ by Huang and Pettie \cite{HuangP18}, which holds even when settling for reachability.

\smallskip 
\noindent\textbf{Spanners and emulators.}  Graph spanners introduced by Peleg and Sch{\"{a}}ffer \cite{PelegS89} are sparse \emph{subgraphs} that preserve shortest path distances up to a small stretch. In contrast to hopsets, those structures are defined only for \emph{undirected} graphs, as no sparsificiation is possible in the worst-case for directed $n$-vertex graphs with $\Omega(n^2)$ edges. 
\begin{defn}[$(\alpha,\beta)$-Spanner]\label{def:spanner}
Given an undirected graph $G=(V,E)$, a subgraph $G^* \subseteq G$ is called a $(\alpha,\beta)$-\emph{spanner} if $\dist_{G^*}(u,v) \leq \alpha \cdot \dist_G(u,v)+\beta$ for every $u, v \in V$.
\end{defn}
An $(\alpha,\beta)$-\emph{emulator} is a weighted set of edges in $V \times V$ (i.e., not necessarily a $G$-subgraph) that provides the same stretch guarantees as $(\alpha,\beta)$-spanners. The four decade old history of spanners has initiated with the earlier works of \cite{PelegS89,AlthoferDDJ90} on multiplicative spanners (where $\beta=0$). Alth{\"{o}}fer et al. \cite{AlthoferDDJ90} showed that a simple greedy procedure yeilds a $(2k-1)$ multiplicative spanner with $O(n^{1+1/k})$ edges. This tradeoff is believed to be optimal by the Erd\H{o}s girth conjecture \cite{erdos1965some}. 
Constructions of spanners with additive stretch of $+2,4,6$ and  $\widetilde{O}(n^{3/2}),\widetilde{O}(n^{7/5}),\widetilde{O}(n^{4/3})$ edges, respectively, are known by \cite{AingworthCIM99},\cite{Chechik13},\cite{BaswanaKMP05}.
%
Abboud and Bodwin \cite{AbboudB17} demonstrated that no $n^{4/3-o(1)}$-size additive spanners exist, even for small \emph{polynomial} stretch. 

Elkin and Peleg \cite{ElkinP04} have demonstrated that the $n^{4/3-o(1)}$ barrier can be bypassed when allowing a small \emph{multiplicative} error; these $(1+\epsilon,\beta)$-spanners are denoted as \emph{near-additive} spanners. For every integer $k\geq 1$ and $\epsilon \in (0,1)$, \cite{ElkinP04} obtained $(1+\epsilon,\beta)$ spanners with $O(\beta \cdot n^{1+1/k})$ edges and $\beta=O(\log k/\eps)^{\log k}$. Thorup and Zwick \cite{ThorupZ06} presented an alternative construction for near-additive emulators, which can be converted to near-additive spanners. 

\smallskip 
\noindent\textbf{The Mystery of the Hopests--Spanners Connection.} 
Currently, near-additive $(1+\epsilon,\beta)$ spanners and near-exact $(\beta,\epsilon)$-hopsets share almost the same tradeoff functions between their size, $\beta$ and $\epsilon$. A collection of recent work noted that not only the combinatorial bounds obtained for these structures are similar but so are also their proof techniques: The Thorup-Zwick algorithm \cite{ThorupZ06}, for example, can be slightly adapted to provide near-additive spanners, emulators and hopsets, 
\cite{HuangP19,ElkinN19}. Abboud, Bodwin and Pettie \cite{AbboudBP18} showed that the TZ bounds are nearly tight for emulators and hopsets (but still quite far from optimal for spanners). 
%
This phenomenon is quite extra-ordinary if one takes into account their fundamental differences:
\begin{itemize}
\item The meaning of their $\beta$ parameter is very different; $\beta$ is the additive \emph{error} term for spanners and the bound on the number of hops for hopsets. \vspace{-4pt}

\item While $(1+\epsilon,\beta)$ spanners apply only for \textbf{undirected, unweighted} graphs\footnote{For weighted graphs, the $\beta$ parameter must depend on the edge weights, see \cite{EGNarXiv19}.}, there are near-exact hopsets for \textbf{undirected, weighted} with the same tradfeoff as given for the unweighted setting. 
In addition, sub-quadratic hopset constructions can also be provided for \textbf{directed, weighted} graphs. (The latter is impossible, in the worst case, for spanners). \vspace{-4pt}

\item For near-additive spanners it is sufficient (and necessary) to restrict attention to \textbf{nearby} pairs (i.e., pairs at distance $O(\beta/\epsilon)$ in $G$). Hopset constructions for unweighted graphs, on the other hand, are concerned with \textbf{distant} pairs (e.g., at distance $\Omega(\beta)$). \vspace{-4pt}

\item Spanners are trivial when the input graph is sparse, but this is not necessarily true for hopsets. 
For example, for a linear-size graph $G$ and parameters $\epsilon=0, \beta=o(n^{1/5})$, the output $(1+\epsilon,\beta)$ spanner is simply $G$. In contrast, the output $(\beta,\epsilon)$-hopset for $G$ might require (in the worst case) a super-linear number of edges.
\end{itemize}

This quite mysterious connection between hopsets and spanners has been a subject for a thorough research \cite{ElkinN16,ElkinN17,AbboudBP18,HuangP19,ElkinN19,Ben-LevyP20,ShabatS21}. In their inspiring and comprehensive survey on this topic \cite{ENSuvery20}, Elkin and Neiman conclude by asking:
\begin{quote}[\cite{ENSuvery20}]
\emph{There is a striking similarity not just between the results concerning near-additive spanners
for unweighted graphs and near-exact hopsets for weighted ones, but also between the techniques used to
construct them and to analyze these constructions. A very interesting open problem is to explain the relationship between near-additive spanners and near exact hopsets rigorously, i.e., by providing a reduction between these two objects.}
\end{quote}

\smallskip 

\noindent\textbf{Distance Preservers.} An additional very active line of graph sparsification research considers the fundamental property of pairwise shortest path distances. For a given graph $G=(V,E)$ and a subset of $p$ demand pairs $P \subset V \times V$, a distance preservers is a sparse subgraph $G^* \subseteq G$ satisfying that $\dist_{G^*}(u,v)=\dist_G(u,v)$ for every $u,v \in P$. Distance preservers, introduced by Coppersmith and Elkin \cite{CoppersmithE06} have played an important role in the constructions of spanners \cite{Pettie09,ElkinFN17}, distance oracles \cite{ElkinP16}, as well as in establishing lower bounds for spanners \cite{AbboudB16} and hopsets \cite{HuangP18,AbboudB18}. The related notion of pairwise spanners allows one to approximately preserve the pairwise distances, up to a small stretch. Pairwise preservers that admit a multiplicative stretch of $(1+\epsilon)$ for $\epsilon \in (0,1)$ are denoted as \emph{near-exact preservers}.  
There is a long line of works on exact preservers \cite{BollobasCE03,CoppersmithE06,BodwinW16,Bodwin17,Bodwin21} and additive pairwise spanners \cite{CyganGK13,KavithaV13,Kavitha15,KavithaV15,AbboudB16} from an upper bound and lower bound perspectives. In their seminal work,  Coppersmith and Elkin \cite{CoppersmithE06} provided a family of lower bounds for exact preservers of $\Omega(n^{2d/(d^2+1)}\cdot p^{d(d-1)/d^2+1})$ edges for all integers $d \geq 2$. Hence, exact preservers might be very dense in the worst-case.

Near-exact preservers are known to be quite sparse but only for unweighted undirected graphs. Using a folklore construction (see e.g., \cite{AhmedBSHJKS20}) $(1+\epsilon,\beta)$ spanners can be converted into near-exact preservers with $O(n+n^{o(1)}p)$ edges. No analogous result is known for weighted graphs (as near-additive spanners exist only in the unweighted detting). In this paper, we fill in this missing gap by providing an alternative reduction from (near-exact) hopsets to (near-exact) preservers. Unlike spanners, the TZ hopset have the same quality in the weighted case, which allows us to compute sparse (near-exact) weighted preservers that almost match their unweighted counterparts. The latter has immediate applications to additional useful notions of spanners, such as sourcewise spanners \cite{RodittyTZ05,Parter14} and spanners with slack \cite{ChanDG06}. We next present our main contribution in more details.

\vspace{-5pt}\subsection{New Results}\label{sec:contribution} \vspace{-5pt}

Throughout, the notation $\widehat{O}(.)$ (resp., $\widetilde{O}(.)$) hides $n^{o(1)}$ (resp. $\poly(\log n)$) factors. 
For clarity of presentation, we assume that the edge weights of the weighted graph considered are polynomial. 

\smallskip

\noindent \textbf{From Hopset to Distance Preservers.}
We translate constructions of near-exact hopsets into near-exact preservers in a black-box manner. Importantly, we do not adapt existing hopset algorithms into preserver algorithms, but rather only convert their \emph{output} hopsets into near-exact preservers. While the reduction scheme is general, in the context of directed graphs, it is most insightful to consider the family of algorithms that compute for every $\beta\leq n^b$ (for some given parameter $b \in [0,1)$) and any $\epsilon \in (0,1)$, a $(\beta,\epsilon)$ (directed) hopset with $\widetilde{O}(n^{2}/\beta^a)$ edges, for some parameter $a > 1$. The intuition behind this function is the following. First, the size bound converges to $n^2$ as $\beta$ approaches to $1$ (which indeed makes sense as for $\beta=1$, the hopset collides with the transitive closure of the graphs). Second, this function captures all existing hopset constructions for directed graphs: In the folklore algorithm of \cite{UllmanY91}, $a=2$ and $b=0$. In the algorithm of \cite{KoganParter22}, we have $a=3$ and $b=1/4$.  

\BTHM[\textbf{From Directed Hopset Hierarchy to Preservers}] \label{directed_preserver_mainthm1}
Let $0 \leq  b <\min(1,1/(a-1))$ and $a >1$ be fixed parameters. Then, given any algorithm for computing $(\beta,\epsilon)$ hopsets for $n$-vertex directed (possibly weighted) graphs with $\widetilde{O}_{\epsilon}(n^{2}/\beta^a)$ edges for every $\beta \leq n^b$, the following holds:
For any $n$-vertex directed (possibly weighted) graph $G$ and a set $P$ of $p$ demand pairs, one can compute an $(1+\epsilon')$-distance preserver $G^* \subseteq G$ where 
$$\epsilon'=O(\epsilon \log\log n) \mbox{~and~} |G^*|=\widetilde{O}_{\epsilon'} (n \cdot p^{1-1/k} + n^{2/a} \cdot p^{1-1/a}) \mbox{~~for~} k=\frac{2-a\cdot b}{1-b}~.$$
\ETHM

More specifically, Theorem \ref{directed_preserver_mainthm1} is based on translating a \emph{collection} (or \emph{hierarchy}) of $(\beta_i,\epsilon)$ hopsets, for $\beta_1 > \beta_2 > \ldots >\beta_\ell$ where $\ell=O(\log\log n)$, into near-exact preservers. The sparsity of the preservers is a function of the size of the given hopsets and their $\beta$ parameter. Note that Lemma \ref{directed_preserver_mainthm1} can be safely applied even for $\epsilon>1/\log\log n$ (e.g., in our applications to reachability preservers). Plugging the bounds of $(\beta,\eps)$ hopsets from \cite{KoganParter22} (i.e., $a=3$ and $b=1/4$) to Theorem \ref{directed_preserver_mainthm1} yields the first construction of near-exact hopsets for weighted and directed graphs:

\begin{thm}[$(1+\epsilon)$ \textbf{Directed Preservers}]\label{lem:directed-preservers}
For every $n$-vertex \emph{weighted directed} graph $G=(V,E)$ and $p$ pairs $P \subset V \times V$, one can compute a $(1+\epsilon)$-preserver for $P$ with $\widetilde{O}(np^{2/5}+(np)^{2/3})$ edges.
\end{thm}
So-far, only \emph{exact} preservers have been known for this setting with $O(n\sqrt{p})$ edges for unweighted graphs \cite{CoppersmithE06}, and $O(\min\{n\sqrt{p}, n^{2/3}p\})$ edges for weighted graphs \cite{CoppersmithE06,BodwinW16,Bodwin21}.
Interestingly, for $p \geq n^{5/4}$, our bounds match the state-of-the-art results for (the weaker notion of) reachability preservers by Abboud and Bodwin \cite{AbboudB18}. Moreover, it also matches the lower-bound for exact preservers by Bodwin \cite{Bodwin21}. Theorem \ref{directed_preserver_mainthm1}  is quite general and can also be shown to recover other known results in the literature but in a \emph{black-box} manner. For example, applying this theorem with the 
\emph{diameter-reducing shortcut}\footnote{The latter can viewed as $(\beta, n)$ hopsets as it preservers merely reachability.} algorithm of \cite{KoganParter22} for which $a=3$ and $b=1/3$, provides \emph{reachability preservers} with $\widetilde{O}(np^{1/3}+(np)^{2/3})$ edges. Hence, matching the known bound of \cite{AbboudB18} for $p \geq n$. Plugging in this theorem the known bounds for \emph{exact} hopsets, for which $a=2$ and $b=0$, recovers the Coppersmith-Elkin bounds of 
$O(n \cdot \sqrt{p})$ edges for \emph{exact} preservers \cite{CoppersmithE06}.

We also provide a similar reduction for \emph{undirected} graphs for which considerably sparse hopsets exist. 
By using the state-of-the-art bounds for $(\beta,\epsilon)$ hopsets of \cite{HuangP19,ElkinN19} we provide the first near-exact preservers for weighted undirected graphs:

\begin{thm}[$(1+\epsilon)$ \textbf{Undirected Weighted Preservers}]\label{lem:preservers-weighted}
For every $n$-vertex weighted undirected graph $G=(V,E)$ and a subset of $p$ pairs $P$, one can compute a $(1+\epsilon)$-preserver for $P$ with $n^{o(1)}\cdot (n+p)$ edges.
\end{thm}
This almost matches the state-of-the-art bounds of $O(n+p \cdot n^{o(1)})$ edges known for the unweighted case.
Note that while the unweighted case follows easily from near-additive spanners, our weighted preservers cannot be obtained by using the weighted-analogue of near-additive spanners\footnote{This is because the additive term in those spanners depends linearly on the edge weight.} (e.g., as given by \cite{EGNarXiv19}). Instead, Lemma \ref{lem:preservers-weighted} is obtained by using our hierarchy of hopset constructions. 
This hopset-based construction enjoys the fact there are near-exact hopsets that meet the TZ tradeoff even for weighted graphs (in contrast to near-additive spanners).  See Table \ref{table:preservers} for a more detailed comparison with the known bounds.

\begin{table}[th!]
\begin{center}
\begin{tabular}{|l|l|c|l|}
\cline{3-4}
\multicolumn{1}{l}{} && \textbf{Approximation} & \textbf{Upper Bound} \\
\hline
\multirow{4}{*}{\textbf{Undirected}} & \textbf{Unweighted} & \specialcellcenter{$1$ \\ $(1+\epsilon)$} & \specialcellleft{$O(\min\{n + n^{1/2}p, np^{1/3}+n^{2/3}p^{2/3}\})$ \cite{CoppersmithE06,BodwinW16} \\ \noindent $O_{\epsilon}(n+p \cdot n^{o(1)})$ \cite{AhmedBSHJKS20}}\\
\cline{3-4}
& \multirow{2}{*}{\vspace{15pt}\textbf{Weighted}}  & \specialcellcenter{$1$\\ \color{purple} $(1+\epsilon)$} & \specialcellleft{$O(\min\{n + n^{1/2}p, n\sqrt{p}\})$ \cite{CoppersmithE06}\\ \noindent \color{purple} \textbf{New:} $\widetilde{O}_{\epsilon}((n+p) \cdot n^{o(1)})$} \\ 
\hline

\multirow{4}{*}{\textbf{Directed}} & \textbf{Unweighted} & \specialcellcenter{$1$\\$n$ (reachability)} & \specialcellleft{$O\left(\min\{n\sqrt{p}, n^{2/3}p+n\}\right)$ \cite{CoppersmithE06,Bodwin21}\\ \color{blue} $O\left(n + (np)^{2/3}\right)$ \color{black}\cite{AbboudB18}}\\
\cline{3-4}
& \multirow{2}{*}{\vspace{15pt}\textbf{Weighted}}  & \specialcellcenter{$1$\\ \color{purple} $(1+\epsilon)$} & \specialcellleft{$O\left(\min\{n\sqrt{p}, n^{2/3}p+n\}\right)$ \cite{CoppersmithE06,Bodwin21}\\ \color{purple}\textbf{New:} $\widetilde{O}_{\epsilon}\left(np^{2/5}+(np)^{2/3}\right)$}   \\
\hline
\end{tabular}
\end{center}
\label{table:preservers}\caption{The state-of-the-art bounds for pairwise preservers with $p$ pairs. New results provided in this paper are shown in purple. We also obtain reachability preservers (by our black-box reduction) that matches the bound of \cite{AbboudB18} for every $p\geq n$. 
}
\end{table}

\noindent\textbf{From Hopset to Emulators and Spanners.} While the reduction from hopsets to preservers (of Thm. \ref{directed_preserver_mainthm1}) is based on converting a hopset \emph{hierarchy} into the output preservers, we observe that the reduction to emulators for unweighted graphs is significantly simpler and requires only a single hopset\footnote{Hopsets and emulators are indeed expected to be ``closer" as they are both allowed to use non-$G$ edges.}. Specifically, we show that augmenting a near-exact hopset with a multiplicative spanner provides a near-additive emulator! Formally, we have:
\begin{obs}[\textbf{Hopsets} $\mathbb{+}$ \textbf{Multiplicative Spanners} $\mathbb{=}$ \textbf{Emulator}]\label{obs:hopset-to-emulator}
For a given $n$-vertex \textbf{unweighted} graph $G=(V,E)$, let $H^*$ be some $(\beta,\eps)$ hopset for $G$ and let $G^*$ be a $t$-multiplicative spanner for $G$. Then, $H^* \cup G^*$ is a $(1+\epsilon, \beta \cdot t)$ emulator. 
\end{obs}
By using a standard reduction from emulators to near-additive spanners, we get:

\begin{obs}[\textbf{Near-Exact Hopsets} $\mathbb{\to}$ \textbf{Near-Additive Spanner}]\label{obs:hopset-to-spanner}
For an unweighted $n$-vertex graph $G=(V,E)$, let $H^*$ be some $(\beta,\eps)$ hopset for $G$, and let $G^*$ be a $t$-multiplicative spanner for $G$. Then, one can compute an $(1+2\epsilon, \beta \cdot t)$ spanner $\widehat{G}$ with 
$O(t|H^*|\cdot \beta/\epsilon+|G^*|)$ edges. 
\end{obs}

While the near-additive emulators and spanners obtained in this black-box manner have slightly suboptimal bounds (see our \textbf{Concluding Remark.}), we find this reduction to be useful and arguably powerful for two main reasons. (i) It makes a progress on the open problem raised by Elkin and Neiman \cite{ENSuvery20} by providing a near-optimal reduction from hopsets to spanners;  reducing the hopset-type $\beta$ to the spanner-type $\beta$. 
(ii) This reduction serves us later on to provide improved lower bound for \emph{hopsets}. More broadly, it allows one to translate metric-compression lower bounds into hopset lower bounds. 

While the above reduction is limited to unweighted graphs, we use a more delicate reduction for \emph{weighted} graphs. This is done again by converting a hierarchy of hopsets (rather than a single hopset) into a spanner. Using the best-on-shelf undirected hopsets (e.g., of \cite{ElkinN19,HuangP19}) we get:

\begin{thm}[\textbf{New Spanners for Weighted Graphs}]\label{thm:new-spanners}
Given an $n$-vertex \textbf{weighted} graph $G$, the hopsets provided by the algorithms of \cite{ElkinN19,HuangP19} can be converted into the following \emph{subgraphs}:
\begin{enumerate}
\item $(1+\epsilon, O(k/\epsilon)^k \cdot W_{\max} \cdot \log n)$-spanner $G' \subseteq G$ with $\widehat{O}(n^{1+1/(2^{k+1}-1)})$ edges where $W_{\max}$ is the largest edge weight\footnote{See Lemma \ref{lem:hopset-to-weighted-spanner} for the more precise statement.}. 

\item For every integer $k$ and $S \subseteq V$, a $S$-sourcewise $(4k-1+\epsilon)$-spanner $G'\subseteq G$ with $\widehat{O}(n+|S|^{1+1/k})$.

\item For any $\epsilon \in (0,1)$ and integer $k \geq 1$, an $\epsilon$-slack $(12k-1+o(1))$-spanner $G' \subseteq G$ with $\widehat{O}(n+(1/\epsilon)^{1+1/k})$ edges.
\end{enumerate}

\end{thm}
Spanner constructions with weighted additive stretch have been recently presented by \cite{EGNarXiv19,AhmedBSKS20,ElkinGN21,AhmedBSHKS21}. An $S$-\emph{sourcewise} $t$-spanner for a subset of sources $S \subseteq V$ provides $t$-multiplicative stretch for all pairs in $S \times V$ \cite{RodittyTZ05,Parter14}. The state-of-the-art bounds for sourcewise spanners are given by Elkin, Filtser and Neiman \cite{ElkinFN17} that provided $S$-sourcewise $(4k-1)$-spanners with $O(n+\sqrt{n}|S|^{1+1/k})$ edges. The notion of $\epsilon$-slack $t$-spanners introduced by Chan, Dinitz and Gupta \cite{ChanDG06} guarantees a $t$-multiplicative stretch for all but $\epsilon$ fraction of the pairs. Our bounds should then be compared with \cite{ChanDG06} which provide $\epsilon$-slack $(12k-1+o(1))$-spanners with $O(n+(1/\epsilon)^{1+1/k}\sqrt{n})$ edges. 
In the lack of near-exact preservers for weighted graphs, the prior constructions of sourcewise spanners and spanners with slack simply used the exact preservers of Elkin and Coppersmith \cite{CoppersmithE06} with $O(n+\sqrt{n}p)$ edges. 
Our constructions use instead the near-exact preservers of \Cref{lem:preservers-weighted} which have only $\widehat{O}(n+p)$ edges, hence sparser for $p > n^{1/2+o(1)}$. 

\smallskip
\noindent \textbf{New Hopset Lower Bounds.}
Currently all existing lower bound constructions for hopsets are based on packing multiple pairs of vertices whose shortest paths are unique, edge disjoint and long. These extremal graphs have been studied by Alon \cite{Alon02}, Hesse \cite{Hesse03} and Coppersmith and Elkin \cite{CoppersmithE06}. Interestingly, the same lower bound techniques (up to small adaptations) are shared by spanners, emulators and hopsets \cite{Hesse03,HuangP18,AbboudBP18,LWWZ22}. Due to the fundamental differences between hopsets and metric-compression structures, prior work had to repeat the arguments for hopsets while introducing adaptions. We obtain the following lower bounds in black-box manner:
\begin{thm}[\textbf{New Hopset Lower Bounds}]\label{thm:near-exact-LB}
For undirected $n$-vertex graphs, we have: \vspace{-5pt}
\begin{itemize}
\item{\textbf{Exact Hopsets}:} For weighted graphs, exact hopsets of linear size admit a worst-case hopbound of $\beta=\Omega(n^{1/3})$. For undirected unweighted graphs the lower bound becomes $\beta=\Omega(n^{1/5})$.   \vspace{-5pt}

\item{\textbf{Hopsets with Sublinear Error}:} Any hopset with $f(d)\leq d+O(kd^{1-1/k})+\widetilde{O}(1)$ and $\beta =O(kd^{1-1/k})+\widetilde{O}(1)$ requires in the worst case $\Omega(n^{1+1/(2^{k}-1)-o(1)})$ edges. 
\end{itemize}

\end{thm}
Our hopbound lower bound for exact hopsets improves upon the the state-of-the-art $\Omega(n^{1/6})$ lower bound by Huang and Pettie \cite{HuangP18} which holds for reachability. No stronger bounds were known for exact hopsets. 
The lower bounds for hopsets with sublinear error nearly match the emulator size bounds by Abboud, Bodwin and Pettie \cite{AbboudBP18}.

\noindent\textbf{From Reachability Preservers to $d$-Shortcuts Lower Bounds.} Our reduction scheme of Theorem \ref{directed_preserver_mainthm1} also allows us to translate lower bounds for reachability \emph{preservers} \cite{AbboudB18} into a lower bound on the diameter (hopbound) of $d$-shortcuts (i.e., $(\beta=d,\eps=n)$ hopsets) of linear size. We show:
\vspace{-3pt}
\BTHM\label{thm:general_shortcut_bound1}
If there is a lower bound of $\Omega(n^{\alpha} \cdot p^{\gamma})$ edges, where $\alpha+\gamma\geq 1+o(1)$, for reachability preservers with $p$-pairs for $n$-vertex graphs, then there exists an $n$-vertex directed graph for which any linear sized $d$-shortcut must satisfy that $d\geq n^{(\alpha+\gamma -1 - o(1))}$.
\ETHM
\vspace{-4pt}
Plugging in this theorem the state-of-the-art lower bound for reachability preservers of \cite{AbboudB18} yields the state-of-the-art diameter lower bound of \cite{HuangP18}. Any improved lower bound result for reachability preservers would immediately lead to improved diameter bound for shortcuts.

\vspace{-5pt}\subsection{Technical Overview} \vspace{-4pt}

Our key result is a reduction from a hierarchy of hopsets to preservers and spanners in weighted and directed graphs. The reduction works by converting the given hopset hierarchy into an intermediate graph structure that we call \emph{bounded-missing spanners}. This notion of spanners, which to the best of our knowledge has not been introduced before, serves as a key transitional junction on the road from hopsets to distance preserving subgraphs.

\noindent\textbf{New Graph Notion: Bounded-Missing Spanners.} For a given (possibly weighted and directed) graph $G$, a \emph{bounded-missing spanner} is a subgraph of $G$ that contains all but a bounded number of (missing) edges of some approximate $u$-$v$ shortest-path for every $u,v \in V$. For integer parameters $r$ and $t$, a subgraph $G'$ of $G$ is an $r$-missing $t$-spanner if for every $u,v \in V$, there is some $u$-$v$ path $P_{u,v} \subseteq G$ of length at most $t\cdot \dist_G(u,v)$ such that $|P_{u,v}\setminus G'|\leq r$; i.e., for every $u,v \in V$ there is a $t$-approximate shortest path in the spanner that misses at most $r$ edges. We call such a $u$-$v$ path $P_{u,v}$ an $r$-missing $t$-approximate path. Note that for $r=0$, we get the standard definition of $t$-spanners. Our proof technique is based on the following steps. First, we show how to convert a given hierarchy of near-exact hopsets into an $r$-missing $t$-spanner $G' \subseteq G$. This also provides a polynomial time algorithm for computing the desired $t$-approximate $r$-missing shortest paths in $G$ w.r.t $G'$. Second, we show that for a wide class of hopset algorithms (namely, those that admit natural size vs. hopbound tradeoff functions) one can carefully select the hopset hierarchy in a way that optimizes the tradeoff between $r,t$ and the size of the output spanner. Lastly, we show how to convert these missing spanners into approximate distance preservers and near-additive spanners. This last part is the most immediate. We next elaborate mainly on that two first parts.

\smallskip

\noindent \textbf{The Key Sparsification Lemma.} Our main lemma shows that for any given hierarchy of $(\beta_i,\epsilon)$-hopsets $H_0, H_1, H_2, \ldots, H_\ell$, for a decreasing sequence $n=\beta_0>\beta_1 > \beta_2 > \ldots >\beta_\ell$, one can compute a $r$-missing $t$-spanner $G$ with $r=\beta_\ell$ and $t=(1+\epsilon)^{\ell}$ and 
\begin{equation}\label{eq:size-technique}
|G'|=\sum_{i=1}^{\ell} |H_i|\cdot \beta_{i-1} \mbox{~edges~.}
\end{equation}

Note that the subset of $H_i$ graphs are \emph{not} subgraphs of $G$, while the missing spanner $G'$ is a $G$-subgraph. The lemma is shown by applying an inductive argument with the following intuition. Suppose that in the given graph $G$, it holds that $\dist^{(r)}_G(u,v)\leq t\cdot \dist_G(u,v)$ for every $u,v \in V$. In such a case, $G'=\emptyset$ is a valid $r$-missing $t$-spanner for $G$. Hence, computing missing spanners is trivial with this assumption. Recall that reducing the number of hops of approximate shortest paths is precisely the goal of hopsets! Consider next the $(\beta_1,\epsilon)$ hopset $H_1$, namely, the first hopset in the hierarchy. By the hopset definition, for any $u,v \in V$ there is some $(1+\epsilon)$-approximate shortest path $P_{u,v} \subseteq G \cup H$ of at most $\beta_1$ edges. To compute the missing spanner $G'$, each $(x,y)$ edge in $H_1$ is then replaced by some $x$-$y$ shortest path in $G$ (which has at most $\beta_0=n$ edges). This explains the first summand of Eq. (\ref{eq:size-technique}). In step $i\geq 1$ of the induction given the current missing spanner $G'_{i-1}$, it is shown that for each edge $(x,y) \in H_i$ there is already some $(1+\epsilon)^{i-1}$-approximate $x$-$y$ shortest path $P_{x,y}$ in $G$ such that $|P_{x,y} \setminus G'_{i-1}|\leq \beta_{i-1}$. Hence, to include $P_{x,y}$ in $G'$, it is sufficient to add at most $\beta_{i-1}$ edges for each $(x,y) \in H_i$. 
This explains the $i^{th}$ summand of Eq. (\ref{eq:size-technique}). As each inductive step can be shown to increase the shortest-path approximation by $(1+\epsilon)$ factor, we end up with a $\beta_\ell$-missing $(1+\epsilon)^{\ell}$-spanners. 

A-priori, it is unclear if a \emph{nice} hierarchy of hopsets, for which Eq. (\ref{eq:size-technique}) yields a \emph{sparse} output spanner $G'$, even \emph{exists}. Interestingly, we show that it does. This is mainly due to the fact that there are near-exact hopsets of \emph{sublinear} size and with a $o(\sqrt{n})$ hopbound (e.g., of \cite{KoganParter22}). While prior work on hopsets mostly focused on \emph{linear} size hopsets, in our context it is the sublinear regime that plays the critical role in the computation of preserving subgraphs. 

\smallskip

\noindent \textbf{Bounded-Missing Spanners $\to$ Approximate Pairwise Preservers and Spanners.} Consider an $r$-missing $t$-spanner $G' \subseteq G$, and suppose that for every $u,v \in V$ there is a polynomial time algorithm for computing the $r$-missing $t$-approximate $u$-$v$ path $P_{u,v}$ in $G$. It then easy to see that one can obtain an approximate preserver for a given $p$ pairs $P \subset V \times V$ by augmenting $G'$ with the missing edges of $P_{u,v}\setminus G'$ for every $u,v \in P$. This adds at most $p \cdot r$ edges. Taking $t=(1+\epsilon)$ provides a near-exact preservers. Our near-exact preservers find further applications in the context of sourcewise spanners \cite{RodittyTZ05,Parter14,ElkinFN17} and spanners with slack \cite{ChanDG06}. We also show that by augmenting a bounded-missing spanner with a standard multiplicative spanner, one can obtain a near-additive spanner (see Lemma \ref{lem:missing-spanner-to-pands}). 

\smallskip

\noindent \textbf{A High Level Intuition for Theorem \ref{directed_preserver_mainthm1}, Reachability Preservers as a Case Study.} To get a clean intuition into this general reduction, consider\footnote{The main applicability of this theorem is in translating near-exact hopsets to preservers. However, it is easier to convey the ideas on shortcuts, which have a simpler size vs. diameter tradeoff function.} the most basic notion of ``reachability" hopsets, known as \emph{diameter-reducing shortcuts} \cite{Thorup92}. For a given directed graph $G$ with a transitive closure $TC(G)$, a graph $H \subseteq TC(G)$ is a $d$-shortcut if for every $(u,v)\in TC(G)$, $H$ contains a $d$-hop directed path from $u$ to $v$\footnote{A $d$-shortcut is analogous to $(\beta=d,n)$ hopsets, but usually one uses the term hopsets in the context of approximating shortest paths.}. Theorem \ref{directed_preserver_mainthm1} provides a recipe for translating a hierarchy of $d$-shortcuts into \emph{reachability preservers}, namely, subgraphs that preserve reachability between a given set $P$ of $p$ demand pairs \cite{AbboudBP18}. 
We first define the hopset hierarchy (or a $d$-\emph{shortcut} hierarchy), and then show how to translate it into a preserver with $\widetilde{O}(np^{1/3}+(np)^{2/3})$ edges. To define the hopset hierarchy, we use (in a black-box manner) the state-of-the-art bounds for $d$-shortcuts obtained by Kogan and Parter \cite{KoganParter22}: (i) For every $d\leq n^{1/3}$, one can compute a $d$-shortcut with $\widetilde{O}(n^2/d^3)$ edges, and (ii) for every $d> n^{1/3}$, there is a $d$-shortcut with $\widetilde{O}((n/d)^{3/2})$ edges (Lemma \ref{lem:from-super-to-sublinear-hopsets} shows how to derive (ii) from (i)\footnote{This was shown in \cite{KoganParter22} for a particular function, and Lemma \ref{lem:from-super-to-sublinear-hopsets} considers a general function.}). Hence, we use Theorem \ref{directed_preserver_mainthm1} with $a=3$ and $b=1/3$. The proof has two parts, depending on the number of pairs $p$. 

Assume first that $p\leq n$, in which case the dominating term in the desired size of the output preserver is $np^{1/3}$.
We define a hopset hierarchy $H_0,H_1,\ldots, H_\ell$ for $\ell=O(\log\log n)$ where each $H_i$ is a $(\beta_i,\epsilon)$ hopset for an exponentially decaying sequence $n=\beta_0 > \beta_1 \ldots > \beta_\ell=n^{1/3}$. Hence, all hopsets in this hierarchy are of \emph{sublinear} size. We set the final value of $\beta_\ell$ to the value $d$ for which the following equality holds: $\mathbf{(n/d)^{3/2} \cdot d = p\cdot d}$. 
The reasoning behind this equation is as follows. The left term corresponds to the size\footnote{The series of Eq. \ref{eq:size-technique} converges at first point for which $\beta_i=\beta_{i+1}$.} of the $d$-missing $n$-spanner $G'$ (see Eq. (\ref{eq:size-technique})). The right term corresponds to the number of edges needed to be added to $G'$ in order preserve the reachability for the given $p$ pairs (i.e., adding to $G'$, the $d$ missing edges for each pair). Solving this equation for $d$ provides the desired size bound of $O(np^{1/3})$. 

Next, assume that $p> n$, and thus the size of the output preserver is dominated by $\widetilde{O}((np)^{2/3})$. The hopset hierarchy in this case consists of $2\ell$ hopsets. The first $\ell$ hopsets in the hierarchy are $(\beta_i,\epsilon)$ hopsets for $\beta\leq n^{1/3}$ (hence of \emph{sublinear} size, except for the last one). The second set of hopsets are \emph{superlinear} with $\beta_i> n^{1/3}$. We set the final value of $\beta_\ell > n^{1/3}$ to the value $d$ for which the following equality holds: 
$\mathbf{n^2/d^3 \cdot d = p\cdot d}$, which follows the same logic as explained above, with the only distinction that we use the super-linear regime of the $d$-shortcut function. Solving for $d$, provides the desired bound of $\widetilde{O}((np)^{2/3})$. An important conclusion from these relations (also discussed next and in Sec. \ref{sec:LB}) is that lower bounds for reachability preservers would immediately provide lower bounds for \emph{diameter-reducing shortcuts} (see \Cref{thm:general_shortcut_bound1}).

The near-exact weighted and directed preservers of Theorem \ref{lem:directed-preservers} are obtained using the output $(\beta,\epsilon)$-hopsets of \cite{KoganParter22} (for which we plug $a=3$ and $\beta=1/4$ in Theorem \ref{directed_preserver_mainthm1}, see Thm. \ref{existential_hopset_thm1}). 

\smallskip

\noindent \textbf{Graph Compression Lower Bounds $\to$ Hopset Lower Bounds.} Finally, we enjoy the complementary aspect of our reduction to provide lower bounds for hopsets. These lower bounds are obtained in a black-box manner from existing lower bound results for preservers, spanners and emulators.  Our automatic translation provides two new meaningful results. By translating Bodwin's lower bound result for exact preservers \cite{Bodwin21} we get an improved lower bound on the hopbound $\beta$ for linear-sized hopsets already for undirected graphs. So-far, the only known lower bound for this setting was $\beta=\Omega(n^{1/6})$ by Huang and Pettie (which holds for the weaker notion of diameter-reducing shortcut). We show that Bodwin's lower bounds \cite{Bodwin21} directly improve the hopbound lower bound to $\Omega(n^{1/3})$ and $\Omega(n^{1/5})$ for weighted (reps., unweighted) graphs.  In addition, by translating the lower bound results for emulators of Abboud, Bodwin and Pettie \cite{AbboudBP18}, we get new lower bounds for hopsets with sublinear stretch and hopbound functions. 

A summary of our results and their connections is illustrated in Fig. \ref{fig:summary}.

\begin{figure}[t!]
	\centering

	\begin{tikzpicture}
	[
	every node/.style={
		align = center,
		fill          =  black!5,
		inner sep     = 6pt,
		minimum width = 22pt,
		rounded corners=.1cm}
	]
	
	\node (hopset) at (-11,2) {Near-Exact Hopsets};
	\node (missingsp) at (-3,2) {$r$-Missing $t$-Spanners\\
	(Sec. \ref{sec:missing})};
	\node (pairwisep) at (-3,-1) {(Approximate) Pairwise Preservers\\
	(\Cref{sec:directed-preservers} and \ref{sec:weighted-undirected-hopsets})};
	\node (nearaddspanners) at (-3,5) {Near-Additive Spanners\\
	(\Cref{sec:undirected})};
	\node (slack) at (-6,-4) {Spanners with Slack\\
	(\Cref{sec:weighted-undirected-hopsets})};
		\node (sourcewise) at (0,-4) {Sourcewise Spanners\\
	(\Cref{sec:weighted-undirected-hopsets})};
	
	\draw[->, line width=2.5, green] (hopset) -- (missingsp) node[midway,below,fill=transparent!0]{\textcolor{black}{\Cref{lem:correspondence},\Cref{thm:partial-spanner-undirected}}};
	
	\draw[->, line width=1.5, green] (missingsp) -- 
	(pairwisep) node[midway,right,fill=transparent!0]{\textcolor{black}{\Cref{directed_preserver_mainthm1},\Cref{lem:directed-preservers},\ref{lem:preservers-weighted}, \Cref{lem:missing-spanner-to-pands} and \ref{lem:hopset_to_preserver_undirected_weighted}}};
	
		\draw[->, line width=1.5, green] (missingsp) -- 
	(nearaddspanners) node[midway,right,fill=transparent!0]{\textcolor{black}{\Cref{lem:missing-spanner-to-pands} and \ref{lem:hopset-to-weighted-spanner}}};
	
		\draw[->, line width=1.5, red, dashed] (pairwisep) -- 
	(hopset) node[midway,left,fill=transparent!0]{\textcolor{black}{\Cref{thm:exact-hopset-wu},\ref{thm:exact-hopset-uu}}}; 
	
		\draw[->, line width=1.5, red, dashed] (nearaddspanners) -- 
	(hopset) node[midway,left,fill=transparent!0]{\textcolor{black}{\Cref{lem:near-exact-LBhopset},\ref{thm:lbemulator-to-hopset}}}; 
	
			\draw[->, line width=1.5, green] (pairwisep) -- 
	(slack)  node[midway,left,fill=transparent!0]{\textcolor{black}{\Cref{lem:slack-improved} and \Cref{cor:slack-final}}};
	\draw[->, line width=1.5, green] (pairwisep) -- 
	(sourcewise) node[midway,right,fill=transparent!0]{\textcolor{black}{\Cref{thm:onlySbound1} and \Cref{thm:Sandkbound1}}};

	\end{tikzpicture}

	\caption{High-level description of the main results and their connections.}
	\label{fig:summary}
\end{figure}
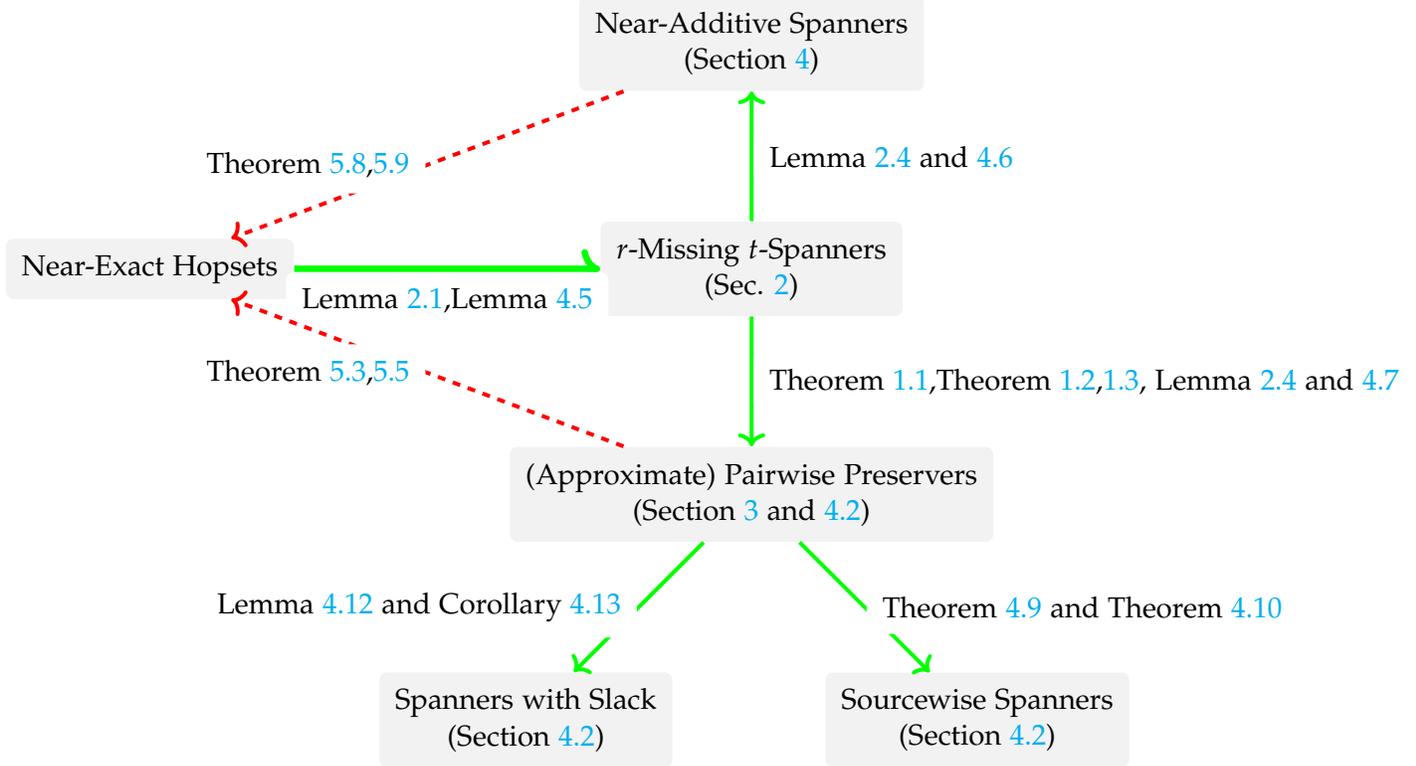

\noindent \textbf{A Concluding Remark.} In this work, we delved into the mysterious connection between hopsets and distance preserving subgraphs by providing a general reduction from the former to the latter. We note that a reverse reduction is somewhat less plausible, as the subgraph problem is trivial for sparse graphs, while hopsets are not. The main two benefits of hopsets (in comparison to spanners) that we enjoy of in our constructions are: (i) there are sublinear hopsets with $o(\sqrt{n})$ hopbounds, and (ii) there are efficient near-exact hopsets for weighted and directed graphs.  Note that while most (if not all) of the known algorithmic applications of hopsets (e.g., for shortest-path computation) require linear size, our work put the spotlight on hopsets of \emph{sublinear} size. 

Finally, we note that our reductions from near-exact hopsets to near-additive emulators and spanners are slightly suboptimal in an ``unavoidable" manner: emulators and spanner do not ``account" for edges incident to low-degree vertices (as one can simply add all these edges to the structure)\footnote{E.g., for a desired size bound of $O(n^{1+1/k})$ edges, one can add all edges incident to vertices with degree at most $n^{1/k}$ to the output subgraph.}. Hopsets however do account for these edges in their hopbound measure\footnote{This leads to the tradeoff differences between emulators and hopsets \cite{HuangP19} in the TZ algorithm \cite{ThorupZ06}.}. While one can slightly adapt the hopset definition (e.g., counting only edges incident to high-degree vertices in the hopbound measure) to make these notions closer, this work seeks for a \emph{black-box connection} with no adaptations! Perhaps surprisingly this on its own  provides a collection of improved subgraph constructions and lower bounds. It would be interesting to study the ``minimal" adaptations for the hopset definition to provide a (hopefully tight) reverse reduction.


%
%

\subsection{Preliminaries}
We use $\widetilde{O}_{\epsilon,W}(.)$ to hide polynomially factor in $1/\epsilon$ and logarithmic factors in the maximum edge weight. For an $n$-vertex directed graph $G$, let $TC(G)$ denote the transitive closure of $G$, and let $TC_W(G)$ denote the transitive closure of $G$ weighted by the corresponding shortest path distances. 
For a possibly weighted graph $G$, let $\dist_G(u,v)$ at the \emph{weight} of the shortest path from $u$ to $v$. 
Let $\len(Q)$ be the \emph{length} of a path $Q$, measured by the sum of its weighted edges. Let $|Q|$ be the number of edges on this path. For unweighted graphs, $\len(Q)=|Q|$. 
For any vertex pair $u,v \in V$, define $\dist_G^{(\beta)}(u,v)$ to be the minimum length $u$-$v$ path with at most $\beta$ edges (hops). If there is no such path, then $\dist_G^{(\beta)}(u,v)=\infty$. We assume w.l.o.g that every edge $(u,v) \in G$ is a shortest path between $u$ and $v$ in $G$. When considering a weighted graph $G=(V,E,W)$ we may omit the weight function $W$, when it is not explicitly used.

\begin{defn}[Pairwise Spanners]\label{def:pairwise-spanner}
For a given (possibly weighted) graph $G=(V,E)$ and a subset of pairs $P$, a subgraph $G^* \subseteq G$ is an $(\alpha,\beta)$ $P$-spanner if $\dist_{G^*}(u,v)\leq \alpha \cdot \dist_{G}(u,v)+\beta$. When $P \subseteq S \times V$ (resp., $P \subseteq S \times S$) for $S \subseteq V$, $G^*$ is denoted as an $(\alpha,\beta)$ sourcewise spanner (resp., subsetwise spanner). 
\end{defn}

%
\begin{lem}[Lemma 5 of \cite{Zwick02,ChechikC20}]\label{lem:bounded-APSP}
Given a directed $n$-vertex graph $G$ with integer weights\footnote{One can also handle rational weights with $O(\log n)$ precision by multiplying each weight by the lowest common denominator of all the edge weights. The running time is polynomial, which is sufficient for our purposes.} in $[1,M]$, there is an algorithm $\APSP^{\leq R}$ that compute the $R$-hop distances and paths in time $O(M \cdot R \cdot n^{\omega})$ time. The output of the algorithm is given by $\mathcal{P}$ that consists of all $u$-$v$ $R$-hop shortest paths $P_{u,v}$ for every $(u,v) \in TC(G)$. I.e., for every $P_{u,v} \in \mathcal{P}$, it holds that $\len(P_{u,v})=\dist^{(R)}_G(u,v)$. 
\end{lem}

\begin{lem}\label{lem:multiplicative-spanner}\cite{AlthoferDDJ90}
For every $n$-vertex (possibly weighted) graph $G$ and a given integer $k \geq 1$, one can compute a $(2k-1)$-spanner $H \subseteq G$ with $|H| \leq n^{1+1/k}$ edges. 
\end{lem}

\begin{inequality}\label{ineq:eps} 
For every $\epsilon \in (0,1)$ and any positive integer $t$, it holds $(1+\epsilon/(2t))^t \leq (1+\epsilon)$. 
\end{inequality}

\smallskip

\noindent\textbf{Diameter-Reducing Shortcuts, Reachability Preservers and Directed Hopsets.}
For a directed graph $G$, a $d$-shortcut $H \subseteq TC(G)$ satisfies that the directed diameter of $H \cup G$ is at most $d$. I.e., for every $(u,v)\in TC(G)$, the graph $G \cup H$ contains a $u$-$v$ path with at most $d$ edges. Note that a $d$-shortcut is simply a $(\beta=d,n)$ hopset. A subgraph $G' \subseteq G$ is a \emph{reachability preserver} for a pair set $P \subseteq V$ if it contains a $u$-$v$ (directed) path for every $(u,v)\in TC(G)\cap P$. The following results are given in  
\cite{KoganParter22} for $d$-shortcuts and near-exact directed hopset:

\BTHM\label{main_theorem_shortcut_paper1}[Theorem 1.1 in \cite{KoganParter22}]
For every $n$-vertex graph $G$, there is an $\tilde{O}(n^\omega)$-time\footnote{Where $\omega$ is the optimal matrix multiplication constant.} randomized algorithm for computing a $d$-shortcut set of cardinality:
\[
S(n,d) = 
\begin{cases}
\widetilde{O}\left( n^2/d^3\right), &\text{for $d \leq n^{1/3}$;} \\
\widetilde{O}( (n/d)^{3/2}),  &\text{for $d > n^{1/3}$,}
\end{cases}
\] 
\ETHM

\BTHM\label{existential_hopset_thm1}[Theorem 1.4 in \cite{KoganParter22}]
For every $n$-vertex graph $G$ with integer weights $[1,M]$, $\epsilon \in (0,1)$ and $\beta \in \mathbb{N}_{\geq 1}$, there is an $\widetilde{O}(n^3 \cdot \poly\log(n M))$-time randomized algorithm for computing $(\epsilon,\beta)$ hopsets whose number of edges is bounded by:
\[
H(n,\beta,\epsilon) = 
\begin{cases}
\widetilde{O}_{W,\epsilon}\left( \left(\frac{n^2}{\beta^3}\right) \right),  &\text{for $\beta \leq n^{1/4}$;} \\
\widetilde{O}_{W,\epsilon}\left( \left(\frac{n}{\beta}\right)^{5/3} \right), &\text{for $\beta > n^{1/4}$}
\end{cases}
\]
\ETHM


\vspace{-8pt}
\section{The Key Reduction: Hopsets $\mathbf{\to}$ Missing Spanners}\label{sec:missing}\vspace{-3pt}

Let $G =(V, E)$ be an $n$-vertex (possibly weighted and directed) graph. 

\begin{defn}[$r$-missing $t$-spanner]\label{def:missing-spanners}
A subset of edges $G' \subseteq G$ is an $r$-missing $t$-spanner for $G$ if for every $(u,v) \in TC(G)$, there is a $u$-$v$ path $P_{u,v} \subseteq G$ satisfying:
\begin{equation}\label{eq:missing-approx-path}
\len(P_{u,v})\leq t \cdot \dist_G(u,v) \mbox{~and~} |P_{u,v} \setminus G'| \leq r~.
\end{equation}
For $r=0$, a $0$-missing $t$-spanner is simply a $t$-spanner.
\end{defn}


Our key lemma shows that one can transform a family of $(\beta_i,\epsilon)$-hopsets for a sequences $\beta_1 \geq \beta_2 \geq \ldots \geq \beta_\ell$ into an $r$-missing $t$-spanner for $r=\beta_\ell$ and $t=(1+\epsilon)^{\ell}$. 
This lemma is general enough to later on provide black-box constructions of near-exact preservers (even in the directed weighted setting), spanners, as well as, reachability preservers (for which $\epsilon=n$). 

\begin{lem}\label{lem:correspondence}[From Hopsets to $r$-Missing $t$-Spanners]
Let $G = (V,E,\omega)$ be an $n$-vertex (possibly directed) graph and a set $P \subseteq V \times V$ of demand pairs. Then given a sequence of $\ell$ values $n=\beta_0 \geq \beta_1 \geq \beta_2 \geq \ldots \geq \beta_\ell$ and a 
a collection of $(\beta_i,\epsilon)$-hopsets $H_i$ for every $i \in \{1,\ldots, \ell\}$, one can compute an $r$-missing $t$-spanner $G'$ for $r=\beta_\ell$ and $t=(1+\epsilon)^{\ell}$ of cardinality
$|G'|\leq \sum_{i=1}^{\ell} |H_i| \cdot \beta_{i-1}$ (where $|G'|$ denotes the number of edges in the spanner $G'$).
\end{lem}


We start by describing an algorithm for computing the missing spanner $G'$ given the collection of hopsets $H_1,\ldots, H_\ell$. 
%
Alg. $\HopsetToPartialSpanner$ builds the missing spanner $G'$ in $\ell$ consecutive steps. Let $G'_0, H_0=\emptyset$. Then the output of each step $i\geq 1$ is a subgraph $G'_i$ computed as follows. Let $\mathcal{P}'_i$ be collection of all $V \times V$ $(\beta_{i-1})$-bounded hop shortest paths in the graph $G \cup H_{i-1}$ (where for each pair of vertices we take exactly one path arbitrarily). These paths can be computed by applying Alg. $\APSP^{(\beta_{i-1})}(G \cup H_{i-1})$ of Lemma \ref{lem:bounded-APSP}. Then, for every edge $(u,v) \in H_i$ the algorithm adds to $G'_i$ the edges $P_{u,v} \cap G$ where $P_{u,v}$ is the $\beta_{i-1}$-hop $u$-$v$ path in $\mathcal{P}'_i$. Note that as $P_{u,v}$ is computed in the augmented graph $G \cup H_{i-1}$, the algorithm adds to $G'_i$ only its $G$-edges. The output missing spanner is the output of the $\ell^{th}$ step, $G'_\ell$. 

\begin{mdframed}[hidealllines=false,backgroundcolor=gray!00]
\center \textbf{Algorithm $\HopsetToPartialSpanner$}
\begin{flushleft}
\textbf{Input:} An $n$-vertex $G=(V,E,\omega)$, hopset hierarchy of $(\beta_i,1+\epsilon)$-hopsets $H_i$, $i \in \{1,\ldots, \ell\}$.\\
\textbf{Output:} A $r$-missing $t$-spanner $G'_\ell \subseteq G$ for $r=\beta_\ell$ and $t=(1+\epsilon)^{\ell}$.
\end{flushleft}

\vspace{-8pt}
\begin{enumerate}
\item Set $G',G'_0 = \emptyset$. 
\item For $i=1$ to $\ell$ do:
\begin{enumerate}
    \item $G'_i \gets G'_{i-1}$.
		\item $\mathcal{P}'_i \gets \APSP^{(\beta_{i-1})}(G \cup H_{i-1})$ where $\mathcal{P}'_i=\{ P^i_{u,v} ~\mid~ (u,v) \in TC(G)\}$.
    \item For each edge $(u,v)\in H_i \setminus \bigcup_{j \leq i-1} H_j$ do: 
    \begin{itemize}
      \item $G'_i \gets G'_i \cup (P^i_{u,v} \cap E(G))$.
    \end{itemize}
\end{enumerate}
\item Return $G'_\ell$.
\end{enumerate}
\end{mdframed}

We now turn to analyze the algorithm and start with the size analysis.

\begin{lem}\label{lem:size-partial-preserver}
$|G'_\ell| \leq \sum_{i=1}^{\ell} |H_i| \cdot \beta_{i-1}$.
\end{lem}
\begin{proof}
We claim that in iteration $i$, the algorithm adds to $G'_i$ at most $|H_i| \cdot \beta_{i-1}$ edges to $G'_i$. 
This holds since for every edge $e=(u,v) \in H_i \setminus \bigcup_{j \leq i-1} H_j$, the algorithm computes a $u$-$v$ path of at most $\beta_{i-1}$ hops, therefore adding at most $|H_i| \cdot \beta_{i-1}$ edges. 
\end{proof}

We next show the more delicate argument that $G'_\ell$ is indeed an $\beta_\ell$-missing $(1+\epsilon)^{\ell}$-spanner. 
For ease of presentation, let $H_{\ell+1}=TC(G)$ and $G'_{\ell+1}=G$. We show by induction on $i \in \{1,\ldots, \ell+1\}$, that $G'_{i-1}$ is a (partial) $\beta_{i-1}$-missing $(1+\epsilon)^{i-1}$-spanner of the subgraph $G'_i$ but only w.r.t to $u,v$ pairs such that $(u,v)\in H_i$. For $\ell+1$, since $G'_{\ell+1}=G$ and $H_{\ell+1}=TC(G)$, we get that $G'_\ell$ is indeed a $\beta_\ell$-missing $(1+\epsilon)^{\ell}$-spanner.

\begin{lem}\label{lem:inductive-stretch}
For every $i \in \{1,\ldots, \ell+1\}$, there is a polynomial time algorithm that given $G'_i$ returns a collection of paths $\mathcal{Q}_i=\{Q^i_{u,v} \subseteq G'_i ~\mid~ (u,v) \in H_i\}$ in $G'_i$, such that every path $Q^i_{u,v} \in \mathcal{Q}_i$ satisfies:
\begin{enumerate}
\item $\len(Q^i_{u,v}) \leq (1+\epsilon)^{i-1} \cdot \dist_G(u,v)$. 
\item $|Q^i_{u,v} \setminus G'_{i-1}| \leq \beta_{i-1}$. 
\end{enumerate}
In other words, $G'_{i-1}$ satisfies the $\beta_{i-1}$-missing $(1+\epsilon)^{i-1}$-spanner properties for every pair $(u,v)\in H_i$. 
\end{lem}

\begin{proof}
For ease of notation, let $\mathcal{P}'_{\ell+1} \gets \APSP^{(\beta_{\ell})}(G \cup H_{\ell})$ where $\mathcal{P}'_{\ell+1}=\{ P^{\ell+1}_{u,v} ~\mid~ (u,v) \in TC(G)\}$.
We prove the lemma by induction on $i \in \{1,\ldots, \ell+1\}$. The base case of $i=1$ holds by taking $Q^1_{u,v}$ as a shortest path in $G$. Since $\beta_0=n$, we have that $Q^1_{u,v}=P^1_{u,v} \in \mathcal{P}'_1$ and indeed $Q^1_{u,v} \subseteq G'_1$. Property (2) is immediate as $\beta_0=n$. 

Assume that the claim holds up to $(i-1)$, consider $i \in \{ 2, \ldots, \ell\}$ and an edge $(u,v) \in H_i$. Our goal is to show that $G'_i$ contains a path $Q^i_{u,v}$ that satisfies the two desired properties. We define the path $Q^i_{u,v}$ as follows. 
Let $P^i_{u,v} \in \mathcal{P}'_i$ be the $u$-$v$ shortest path with (at most) $\beta_{i-1}$ hops in $G \cup H_{i-1}$ (obtained in Step (2b)). Since $H_{i-1}$ is an $(\beta_{i-1},\epsilon)$-hopset, it holds that:
%
\begin{equation}\label{eqn:bounded_len_esetimate1}
\len(P^i_{u,v})=\dist^{(\beta_{i-1})}_{G \cup H_{i-1}}(u, v) 
\leq (1+\epsilon)\dist_G(u,v)~.
\end{equation}

Let $P^i_{u,v}=[u=x_1, x_2,\ldots, x_q=v]$ and denote $e_j=(x_j, x_{j+1})$ for every $j \in \{1,\ldots, q-1\}$. For every $e_j$, define a $x_j$-$x_{j+1}$ path $P'_j$, as follows. For $e_j \in E(G)$, let $P'_j=e_j$. Otherwise, it holds that $e_j \in H_{i-1}$, and $P'_j$ is defined by taking $x_j$-$x_{j+1}$ path $Q^{i-1}_{x_j,x_{j+1}}\subseteq G'_{i-1}$ that belongs to the set $\mathcal{Q}_{i-1}$. Since $e_j \in H_{i-1}$, the path $Q^{i-1}_{x_j,x_{j+1}}$ is well defined by the induction assumption. 
Let $Q^i_{u,v}=P'_1 \circ P'_2 \circ \ldots \circ P'_q$. See Fig. \ref{fig:induc-missing} for an illustration of the definition of $Q^i_{u,v}$. 
We now show that $Q^i_{u,v} \subseteq G'_i$. Since $G'_{\ell+1}=G$, this holds for $i=\ell+1$ and it is sufficient to consider $i \in \{2,\ldots, \ell\}$. By the induction assumption for $i-1$, $G'_{i-1}$ contains the path $Q^{i-1}_{x,y}$ for every $(x,y)\in H_{i-1}$. Since at the end of step $i \leq \ell$, the algorithm adds to $G'_i$, the edges in $P_{u,v} \cap E(G)$, and since $G'_{i-1}\subseteq G'$, we conclude that $Q^i_{u,v} \subseteq G'_i$.  We next bound the length of $Q^i_{u,v}$. 
\begin{align}
\len(Q^i_{u,v}) & \leq \sum_{j=1}^q \len(P'_j) \leq  (1+\epsilon)^{i-2} \cdot \sum_{j=1}^q \dist_G(x_j,x_{j+1}) = (1+\epsilon)^{i-2} \cdot \len(P_{u,v})   \notag \\
&= (1+\epsilon)^{i-2} \cdot \dist^{(\beta_{i-1})}_{G \cup H_{i-1}}(u, v) \leq  (1+\epsilon)^{i-1} \cdot \dist_G(u,v),  &&\text{(by Eq. (\ref{eqn:bounded_len_esetimate1}))} \notag 
\end{align}
where the second inequality follows bythe induction assumption for $i-1$. 
This satisfies property (1). It remains to show property (2). Recall that $Q^i_{u,v}=P'_1 \circ P'_2 \circ \ldots \circ P'_q$ where $q \leq \beta_{i-1}$. By the induction assumption for $i-1$, each $P'_j$ is either contained in $G'_{i-1}$ or else corresponds to a $G$-edge. Therefore, $|Q^i_{u,v} \setminus G'_{i-1}| \leq  \beta_{i-1}$, and the induction step holds. 
\end{proof}

\begin{figure}[h!]
\begin{center}
\includegraphics[scale=0.35]{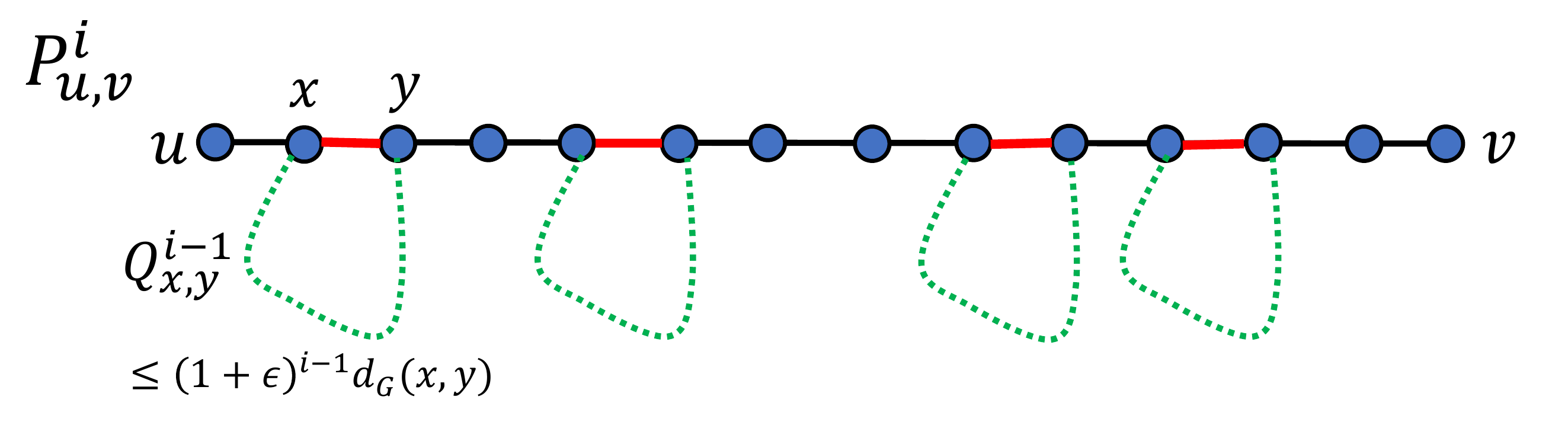}
\caption{\sf An illustration for the inductive argument of Lemma \ref{lem:inductive-stretch}. Shown is the $\beta_{i-1}$-hop $u$-$v$ shortest path $P^i_{u,v} \subseteq G \cup H_{i-1}$ for some hopset edge $(u,v) \in H_{i}$. By the induction assumption, for each $(x,y) \in H_{i-1}$ the subgraph $G'_{i-1}$ contains a $(1+\epsilon)^{i-1}$-approximate shortest path denoted as $Q^{i-1}_{x,y}$. The final path $Q^{i}_{u,v}$ is obtained by replacing each edge $(x,y) \in P_{u,v} \setminus G$ by the path $Q^{i-1}_{x,y}$. This results in a $(1+\epsilon)^i$-approximate $u$-$v$ shortest path which is fully contained in $G'_i$, and at most $\beta_{i-1}$ of its edges are missing from $G'_{i-1}$. 
\label{fig:induc-missing} 
}
\end{center}
\end{figure}

Lemma \ref{lem:correspondence} follows by applying Lemma \ref{lem:inductive-stretch} for $i=\ell+1$ with $G'_{\ell+1}=G$ and $H_{\ell+1}=TC(G)$.  

We conclude this section by showing how missing spanners can be augmented to provide preservers and near-additive spanners.
\begin{lem} \label{lem:missing-spanner-to-pands}
Given an $r$-missing $t$-spanner $G'$ for a possibly directed and weighted graph along with a polynomial time algorithm for computing the $r$-missing $t$-approximate paths in $G$, one can compute a $t$-approximate preserver\footnote{Known also as $t$ pairwise spanners} for a given $p$ pairs $P \subseteq V \times V$ of size $|G'|+p \cdot r$. 

In addition, for undirected and unweighted graphs, one can compute a $(\alpha=t,\beta=r \cdot (2k-1))$ spanner with $|G'|+O(n^{1+1/k})$ edges. 
\end{lem}
\begin{proof}
For every pair $u,v \in P$, let $P_{u,v}$ be the $r$-missing $t$-approximate $u$-$v$ path in $G$. The output preserver $G^*$ is given by $G^*= G' \cup \bigcup_{(u,v)\in P} P_{u,v}\setminus G'$. The size and correctness follow by Def. \ref{def:missing-spanners}.

Now consider the transformation to a near-additive spanner $G^*$. Let $G^*=G' \cup G''$ where $G''$ is a $(2k-1)$ multiplicative spanners (obtained by Lemma \ref{lem:multiplicative-spanner}). The size bound is immediate. We consider the stretch argument for some $u,v\in V \times V$. Let $P_{u,v}$ be the $r$-missing $t$-approximate $u$-$v$ path in $G$. Then for each missing edge $(x,y) \in P_{u,v}\setminus G'$, the multiplicative spanner $G''$ provides a $(2k-1)$-length path. Altogether, we have $\dist_{G^*}(u,v) \leq |P_{u,v} \cap G'|+(2k-1)\cdot |P_{u,v} \setminus G'|\leq t\cdot \dist_{G}(u,v)+r \cdot (2k-1)$ as required. The lemma follows.
\end{proof}
%
%

\vspace{-5pt}
\section{New Distance Preservers for Directed Graphs}\label{sec:directed-preservers}\vspace{-3pt}

In this section, we prove Theorem \ref{directed_preserver_mainthm1} and its concrete applications to directed weighted preservers, Theorem \ref{lem:directed-preservers}. Recall the definition of pairwise spanners (also denoted as approximate distance preservers), see Def. \ref{def:pairwise-spanner}. Exact preservers are $(\alpha,\beta)$ $P$-spanners for $\alpha=1, \beta=0$ (namely, preserve the exact distances). In near-exact preservers, the pairwise distances are preserved up to a multiplicative factor of $\alpha=(1+\epsilon)$, for $\epsilon \in (0,1)$. 
To provide the reduction in its most general form, we consider $(\beta,\epsilon)$ hopsets with a smooth tradeoff function of the following form. There is a threshold hopbound $\beta^*=n^b$
such that for every $\beta \leq \beta^*$ the $(\beta, \epsilon)$-hopset has $\widetilde{O}(n^{2}/\beta^a)$ edges for some $a \geq 1$. In the following, $\epsilon$ can be any arbitrary number (not necessarily in $(0,1)$). This is important in order to also capture reachability\footnote{I.e., translating shortcuts (reachability hopsets) into reachability preservers.} (where $\epsilon=n$). We start by showing that given a hopset algorithm for the superlinear regime, one can derive the sublinear regime. Missing proofs are deferred to Appendix.

\begin{lem}[From Superlinear Hopsets to Sublinear Hopsets]\label{lem:from-super-to-sublinear-hopsets}
Let $0 \leq b <\min(1,1/(a-1))$ and $a > 1$ be fixed parameters. Then, given an algorithm $\mathcal{A}$ for computing $(\beta,\epsilon)$ hopsets for $n$-vertex directed (possibly weighted) graphs with $\widetilde{O}_{W,\epsilon}(n^{2}/\beta^a)$ edges for $\beta = n^b$, there exists an algorithm $\mathcal{A}'$ for computing $(\beta,\epsilon)$ hopsets for $\beta > n^b$ with $\widetilde{O}_{W,\epsilon}((n/\beta)^k)$ edges, where $k=\frac{2-ab}{1-b}~.$
\end{lem}
\BPF 
The algorithm works in a similar manner to the sublinear hopsets provided in \cite{KoganParter22}. On a high level, the algorithm samples a set of landmarks vertices $L$ and computes a graph $N=(L,E')$ where edges in $E'$ are added between landmarks whose shortest path has a sufficiently small number of hops. Alg. $\mathcal{A}'$ applies Alg. $\mathcal{A}$ to compute a $(\beta^*=n^b, \epsilon)$-hopset for the graph $N$. This is shown to also provide the desire hopset for $G$ by setting the parameter carefully as follows. Set 

$$q=\left( \frac{\beta}{(n \cdot \log n)^b} \right)^{1/(b-1)} \mbox{~and~} p = \Theta(q \cdot \log n)~.$$
Let $L=V[p]$ be a random sample of $\Theta(n\cdot p)$ vertices of $G$, obtained by sampling each vertex independently into $L$ with probability of $p$. For a vertex pair $u,v$, let $h(u, v)$ be the minimal number of hops over all the $u$-$v$ shortest path in $G$. I.e., the minimal number $i$ satisfying that $\dist^{(i)}_{G}(u,v)=\dist_{G}(u,v)$. 
Consider a net graph $N=(L,E')$ defined by the edges $E'=\{(u,v) ~\mid~ u,v \in L, h(u,v)\leq q^{-1} \}~.$
Let $\beta' = |L|^b = \Theta((p\cdot n)^b)$.
Then, let $H$ be an $(\beta',\epsilon)$-hopset for $N$ where $\beta' = |L|^b = \Theta((p\cdot n)^b)$. The hopset $H$ is computed by applying Algorithm $\cA$. The output $(\beta,\epsilon)$-hopset for $G$ is then given by $H$. 

We next claim that w.h.p. $H$ is indeed a $(\beta=O(\beta' \cdot q^{-1}), \epsilon)$-hopset for $G$. 
Consider a $u$-$v$ shortest path $P_{u,v}$ in $G$ with $\lceil 10\beta' \cdot q^{-1}\rceil$ hops, and let $u'$ and $v'$ be the closest sampled vertices in $L$ to $u,v$ respectively on $P_{u,v}$. 
Then, w.h.p., $h(u,u'),h(v',v)\leq q^{-1}/2$. 
Let $u'=x_0,x_1, \ldots, x_k=v'$ be the set of sampled vertices in $L$ ordered based on their appearance on $P_{u,v}$. W.h.p., we have that $(x_i,x_{i+1})\in E(N)$ for every $i \in \{0,\ldots, k-1\}$. Therefore, the path $P'=[x_0,x_1, \ldots, x_k] \subseteq N$. Since $H$ is a $(\beta',\epsilon)$-hopset for $N$, it holds that:
$$\dist^{(\beta')}_{N \cup H}(u',v')\leq (1+\epsilon)\len(P')~.$$
As each edge in $N$ corresponds to a $G$-path of $O(q^{-1})$ hops, we get that $G \cup H$ contains a $u$-$v$ path of length at most $(1+\epsilon)\len(P)$ and with at most $\beta=O(\beta' \cdot q^{-1})$ hops.

We now bound the size of $H$. By the guarantees of Alg. $\mathcal{A}$, we have that:
\begin{align}
    |H| &= \widetilde{O}_{W,\epsilon} \left( \frac{|L|^2}{(\beta')^{a}} \right) \notag  \\ 
    &= \widetilde{O}_{W,\epsilon} \left( |L|^{2-a\cdot b} \right) && \text{as $\beta' = |L|^b$} \notag \\
    &=  \widetilde{O}_{W,\epsilon} \left( (p \cdot n)^{2-a\cdot b} \right) && \text{as $|L| = \Theta(n\cdot p)$} \notag \notag \\ 
    &= \widetilde{O}_{W,\epsilon} \left( \left(\frac{n}{\beta} \right)^{\frac{2-a\cdot b}{1-b}} \right) && \text{as $p = q\cdot \log n$ and $q=\left( \frac{\beta}{(n \cdot \log n)^b} \right)^{1/(b-1)}$} \notag \\
    &= \widetilde{O}_{W,\epsilon} \left( \left(\frac{n}{\beta} \right)^{k} \right) && \text{as $k = \frac{2-a\cdot b}{1-b}$} \notag
\end{align}
\EPF
%
%
We are now ready to prove Theorem \ref{directed_preserver_mainthm1}.
\BPF [Proof of Theorem \ref{directed_preserver_mainthm1}]
The proof considers two cases depending on the number of pairs $p$. For a \emph{small} value of $p$ the preservers are obtained by using a sequence of $\ell=O(\log\log n)$ (sublinear) hopsets $H_1,\ldots, H_{\ell}$ obtained by applying Alg. $\mathcal{A}'$ of  Lemma \ref{lem:from-super-to-sublinear-hopsets} and using Lemma \ref{lem:correspondence}. For a \emph{large} value of $p$ the construction is more delicate, it generates first a sequence of $2\ell$ hopsets $H_1,\ldots, H_{2\ell}$. The first half of this set, namely, $H_1,\ldots, H_\ell$ have sublinear number of edges, obtanied by applying Alg. $\mathcal{A}'$. The second half,  $H_{\ell+1},\ldots, H_{2\ell}$ have superlinear number of edges, obtained by applying the hopset algorithm $\mathcal{A}$ for the superlinear regime. We now describe the construction of the hopsets in details. 

\smallskip

\noindent\textbf{Case 1: $p \leq n^{2 - a \cdot b}$}. We start with some preliminaries.
\begin{itemize}
\item Set $D= \frac{n}{p^{1/k}}$ and notice that $D \geq n^b$. Hence $D= n^{\alpha}$ for $\alpha \geq b$.
\item Set $\ell= \lceil \log \log n \rceil$ and $\epsilon = \epsilon' / (2\ell)$. 
\item The hopbound sequence $\beta_1,\ldots, \beta_{\ell}$ is defined by $\beta_i = n^{b_i}$ where 
$b_i = (1-\alpha)(\frac{1}{k})^{i} + \alpha.$
\end{itemize}
Note that $\beta_\ell = O(n^{\alpha}) = O(D)$. By Alg. $\cA'$ of Lemma \ref{lem:from-super-to-sublinear-hopsets}, for every $i \in \{1,\ldots, \ell\}$, one can compute a $(\beta_i,\epsilon)$-hopset $H_i$ for $G$ of cardinality 
$|H_i|=\widetilde{O}_{W,\epsilon} \left( \left(n/\beta_i \right)^k \right)$.

By applying Lemma \ref{lem:correspondence} on the graph $G$ and these $\ell$ hopsets $H_1,\ldots, H_\ell$, we get an 
$r$-missing $t$-spanner $G' \subseteq G$ for $r=\beta_\ell$ and $t=(1+\epsilon)^\ell$ where 
\begin{align} |G'| &\leq \sum_{i=1}^\ell |H_i| \cdot \beta_{i-1} =
 \widetilde{O}_{W,\epsilon}\left(\left(\frac{n}{D}\right)^k  \cdot D \right)=\widetilde{O}_{W,\epsilon'}\left(\left(\frac{n}{D}\right)^k  \cdot D \right)~,
\label{final}
 \end{align}
where the first equality follows from the fact that by the definition of $\beta_i$ we have for all $1 \leq i \leq \ell$,
$\left(n/(\beta_i)\right)^{k} \cdot \beta_{i-1} = \left(n/D\right)^k  \cdot D$,
and the last inequality follows as $\epsilon = \epsilon' / (2\ell)$. 
By the definition of $G'$, we have that for every $(u,v) \in TC(G)$, there is a $u$-$v$ path $P_{u,v} \subseteq G$ that satisfies the following:
\begin{itemize}
\item[(Q1)]  $\len(P_{u,v}) \leq (1+\epsilon)^\ell \cdot \dist_G(u,v) \leq (1+\epsilon') \cdot \dist_G(u,v)$ ( by Inequality (\ref{ineq:eps}) ).
\item[(Q2)] $|P_{u,v} \setminus G'| \leq \beta_\ell = O(D)$. 
\end{itemize}
In addition, there is an polynomial time algorithm for computing these $P_{u,v}$ paths. 
The final distance preserver $G^*$ is given by $G^*= G' \cup \{ P_{u,v} \setminus G' ~\mid~ u,v \in P\}.$
By (Q1), we have that $|G^*|\leq |G'|+ O(D p)$. Using Eq. (\ref{final}), the proof for Case 1 follows. 

\smallskip

\noindent \textbf{Case 2: $p \geq n^{2 - a \cdot b}$}. We start with some preliminaries.
\begin{itemize}
\item Set $D_1= n^b$ and $D_2 = \left(\frac{n^2}{p} \right)^{1/a}$ and notice that $D_2 \leq n^b$. Hence $D_2= n^{\alpha}$ for $\alpha \leq b$.
\item  Set $t= 2\cdot \lceil \log \log n \rceil$ and $\epsilon = \epsilon' / (2t)$.

\item  For each $0 \leq i \leq \ell/2$ let $\beta_i = n^{b_i}$ where $b_i = (1-b)(\frac{1}{k})^{i} + b$.

\item  For each $\ell/2 < i \leq \ell$ let $\beta_i = n^{b_i}$ where $b_i = (b-\alpha)(\frac{1}{a})^{i-\ell/2} + \alpha.$
\end{itemize}
Note that $\beta_{\ell/2} = O(n^{b}) = O(D_1)$ and $\beta_\ell = O(n^{\alpha}) = O(D_2)$.

The sequence of first $\ell/2$ hopsets $H_1,\ldots, H_{\ell/2}$ are obtained by computing a $(\beta_i,\epsilon)$ hopset $H_i$ using Alg. $\mathcal{A}'$. The remaining $\ell/2$ hopsets are obtained by computing a $(\beta_i,\epsilon)$ hopset $H_i$ using Alg. $\mathcal{A}$ for every $i \in \{\ell_2+1,\ldots, \ell\}$. By the size guarantees of Algorithm $\cA'$, we have that $H_i$ is a $(\beta_i,\epsilon)$ hopset with 
$|H_i|=\widetilde{O}_{W,\epsilon} \left( \left(\frac{n}{\beta_i} \right)^k \right)$ edges for every $i \in \{1,\ldots, \ell/2$. Similarly, by Algorithm $\cA$ we have that 
$|H_i|=\widetilde{O}_{W,\epsilon} \left(\frac{n^2}{\beta^a_i} \right)$ for every $i \in \{\ell/2,\ldots, \ell\}$. 
%
By applying Lemma \ref{lem:correspondence} on the graph $G$ with the $\ell$ hopsets $H_i$ for $1 \leq i \leq \ell$, we get an $r$-missing $t$-spanner $G'$ where $r=O(D_2)$ and $t=(1+\epsilon)^{\ell}$ and of cardinality:
 \begin{align} |G'| &\leq \sum_{i=1}^t |H_i(\beta_i,\epsilon)| \cdot \beta_{i-1} 
  \notag \\
 &=
 \widetilde{O}_{W,\epsilon}\left(\left(\frac{n}{D_1}\right)^k  \cdot D_1  + \left(\frac{n^2}{D_2^{a}} \right) \cdot D_2 \right) \label{beta_tower2}\\
 &=  \widetilde{O}_{W,\epsilon'}\left(\left(\frac{n}{D_1}\right)^k  \cdot D_1  + \left(\frac{n^2}{D_2^{a}} \right) \cdot D_2 \right),
 &&\text{as  $\epsilon = \epsilon' / (2t)$} \label{final-second}
 \end{align}
where Identity (\ref{beta_tower2}) follows from the fact that by the definition of $\beta_i$ we have the following: 
 \begin{enumerate}
\item For all $1 \leq i \leq \ell/2$: $\left(\frac{n}{\beta_i}\right)^{k} \cdot \beta_{i-1} = \left(\frac{n}{D_1}\right)^k  \cdot D_1.$
\item For all $\ell/2 < i \leq \ell$: $\left(\frac{n^2}{\beta^a_i}\right) \cdot \beta_{i-1} = \left(\frac{n^2}{D^a_2}\right)  \cdot D_2.$
 \end{enumerate}

For every $u,v \in V$, let $P_{u,v} \subseteq G$ be the $u$-$v$ paths satisfying the properties for $G'$.
Again, the output preserve $G^*$ is obtained by taking $G'$ and adding to $G^*$ the edges $P_{u,v}\setminus G'$ for every $(u,v) \in P$. We have that $|G^*|\leq |G'| +O(D_2 p)$. The size bound follows by combining with Eq. (\ref{final-second}).
%
\EPF

\noindent\textbf{Weighted and Directed Preservers (Theorem \ref{lem:directed-preservers}).}
Theorem \ref{lem:directed-preservers} follows by using the algorithms for $(\beta,\epsilon)$ hopsets of \cite{KoganParter22}described in Theorem \ref{existential_hopset_thm1}. For $a=3$ and $b=1/4$, their algorithm computes $(\beta,\epsilon)$ hopsets for $n$-vertex directed weighted graphs with $\widetilde{O}_{W,\epsilon}(n^{2}/\beta^a)$ edges for $\beta \leq n^b$.
Theorem \ref{lem:directed-preservers} follows by 
plugging these parameters in Theorem \ref{directed_preserver_mainthm1}. 
%
We note that for $p\geq n^{5/4}$ this matches the bounds of reachability preservers\footnote{I.e., $n$-approximate distance preservers.} of \cite{AbboudB18}. It should also be compared against the exact preservers with $O\left(\min\{n\sqrt{p}, n^{2/3}p+n\}\right)$ edges, by \cite{CoppersmithE06,BodwinW16}. 

\smallskip
\noindent\textbf{Additional Implications.}
Finally, we illustrate to the applicability of our reduction in additional settings for which we recover the bounds known in the literature. Let $p$ denote the number of demand pairs throughout. 
\begin{enumerate}
\item In \cite{KoganParter22} it is proven that for parameters $a=3$ and $b=1/3$, there is an algorithm $\mathcal{A}$ for computing $(\beta,n)$ hopsets for $n$-vertex directed unweighted graphs with $\widetilde{O}(n^{2}/\beta^a)$ edges for $\beta \leq n^b$ 
(Theorem \ref{main_theorem_shortcut_paper1}) . Plugging this in Theorem \ref{directed_preserver_mainthm1} provides reachability preservers with
$\widetilde{O}\left(n \cdot p^{1/3} + (n\cdot p)^{2/3} \right)$ edges. For $p \geq n$, this matches the state-of-the-art bounds by \cite{AbboudB18}.  

\item
Plugging parameters $a=2$ and $b=0$ in Theorem \ref{directed_preserver_mainthm1} provides exact preservers with $\widetilde{O}(n p^{1/2})$ edges for directed weighted graph. This matches the Coppersmith-Elkin \cite{CoppersmithE06} bound (see Table \ref{table:preservers}). In this case, the given hopset algorithm simply outputs $H=G$. 
Note that interestingly applying Lemma \ref{lem:from-super-to-sublinear-hopsets} with the parameters $a=2$ and $b=0$ provides the known folklore result (see \cite{UllmanY91}) that there exist an $(\beta,0)$-hopset of size $\widetilde{O}(
n^2/\beta^2)$ for all $1 \leq \beta \leq n$.
\end{enumerate}
\vspace{-8pt}\section{New Spanners and Preservers for Undirected Graphs}\label{sec:undirected}\vspace{-4pt}

In this section we employ our approach to obtain new constructions of spanners and distance preservers in undirected (but possibly weighted) graphs. Subsec. \ref{sec:unweighted-undirected-hopsets} considers the simpler unweighted setting. We show that in this setting a single (rather than a hierarchy) hopset can be converted into a near-additive spanner. Then in Subsec. \ref{sec:weighted-undirected-hopsets} we address the more challenging weighted setting for which (pure) near-additive spanners do not exist. We show that given a hierarchy of hopsets one can compute missing spanners which can then be converted into near-exact preservers for weighted graphs (recall that previously such preservers were known only for unweighted graphs). The missing spanners can also be used to provide a weighted variant of near-additive spanners that have been addressed recently in the literature \cite{EGNarXiv19,AhmedBSKS20,ElkinGN21,AhmedBSHKS21}. 

\vspace{-4pt}
\subsection{Hopsets $\mathbb{\to}$ (Unweighted) Near-Additive Spanners}\label{sec:unweighted-undirected-hopsets}
We start with the following lemma that shows how to convert a near-exact hopset into a near-additive emulator. This reduction has the benefit of preserving the universality of the structure: given a universal hopset, the output emulator is \emph{universal} as well. 
 
\begin{lem}\label{lem:hopset-to-emulator}[Hopsets $\to$ Emulators]
Given an \textbf{unweighted} $n$-vertex graph $G=(V,E)$, a $(\beta,\epsilon)$-hopset $H$, for every integer $k \geq 1$, there is $(1+\epsilon, (2k-1) \cdot\beta)$ emulator $H'$ with $O(|H|+ n^{1+1/k})$ edges. 
\end{lem}
\begin{proof}
The output emulator $H'$ is given by taking the union of $H$ and a $(2k-1)$-spanner for $G$. The size bound follows by Lemma \ref{lem:multiplicative-spanner}, we next bound the stretch for a fixed pair $u,v \in V$.

First observe that a $(\beta,\epsilon)$-hopset $H$ is an $\beta$-missing $(1+\epsilon)$-spanner of the graph $G \cup H$. This holds as by the hopset definition, there is a $u$-$v$ path $P_{u,v}$ in $G \cup H$ satisfying that (i) $\len(P)\leq (1+\epsilon)\dist_G(u,v)$ and (ii) $|P|\leq \beta$. By adding a $(2k-1)$-spanner to $H$ (of \Cref{lem:multiplicative-spanner}), we get that for every edge $(x,y) \in P \cap G$, it holds that $\dist_{H'}(x,y)\leq (2k-1)$. Since all edges of $H$ are in the emulator, and as $P$ consists of at most $\beta$ edges, we get that 
\begin{eqnarray*}
\dist_{H'}(u,v) &\leq& \len(P \cap H) + (2k-1) \len(P \cap G) \leq \len(P)+(2k-1) |P \cap G|
\\& \leq& (1+\epsilon)\dist_G(u,v)+(2k-1)\beta~.
\end{eqnarray*}
The lemma follows.
\end{proof}

By using the standard reduction from near-additive emulators to near-additive spanner (see e.g., \cite{ThorupZ06}), we also have the following, which makes a progress on the open problem raised by Elkin and Neiman \cite{ENSuvery20} and proves \Cref{obs:hopset-to-spanner}.


\begin{cor}\label{lem:hopset-to-unweightedspanner}
Given an unweighted $n$-vertex graph $G=(V,E)$, a $(\beta,\epsilon)$-hopset $H$ and an integer $k \geq 1$, one can compute a $(1+2\epsilon, (2k-1)\beta)$ spanner $G' \subseteq G$ with $O(k \cdot (\beta /\epsilon) \cdot|H|+ n^{1+1/k})$ edges. 
\end{cor}
\begin{proof}
Let $H'=(V, E', W)$ be the $(1+\epsilon, (2k-1)\beta)$-emulator obtained in Lemma \ref{lem:hopset-to-emulator}. The output spanner $G'$ is obtained by replacing every non-$G$ edge $(x,y) \in H'$ of weight at most $10k\beta/\epsilon$, with its corresponding $x$-$y$ shortest path in $G$. Formally, for every $(x,y) \in H' \setminus G$, let $P_{x,y}$ be an $x$-$y$ shortest path in $G$. 
Then, the output spanner $G'$ is given by:
$$G'= (H' \cap G) \cup \bigcup_{(x,y) \in H', W(x,y)\leq 10k\beta/\epsilon} P_{x,y}~.$$
Since we replace $H' \setminus G=H$ edges by $G$-paths of length at most $10k\beta/\epsilon$, the size bound follows. We now provide the stretch argument. Let $\beta'=(2k-1)\beta/\epsilon$.

We start by showing the stretch guarantee for every $u,v \in V$ pair satisfying that $\dist_G(u,v)\leq 2\beta'$. Fix such a $u,v$ pair and let $P'_{u,v}$ be the $u$-$v$ shortest path in $H'$. Since $H'$ is a $(1+\epsilon, (2k-1)\beta)$-emulator, we have that $\len(P'_{u,v}) \leq (1+\epsilon)\dist_{G}(u,v)+(2k-1)\beta < 10k\beta/\epsilon$. 
Since $G'$ contains the path $P_{x,y}$ for every $(x,y) \in P'_{u,v}$, we have that $\dist_{G'}(x,y)=\dist_{H'}(x,y)$ for every $(x,y) \in P'_{u,v}$. Consequently, $\dist_{G'}(u,v)=\dist_{H'}(u,v) \leq (1+\epsilon)\dist_G(u,v)+(2k-1)\beta$. 

We now turn to consider an arbitrary pair $u',v' \in V$ of distance at least $2\beta'$ in $G$. Let $P_{u',v'}$ be the $u'$-$v'$ shortest path in $G$. Partition the path $P_{u',v'}$ into consecutive segments of length $[\beta',2\beta']$. Let these segments be given by $P_{u',v'}=P_1 \circ P_2 \circ \ldots \circ P_\ell$ and let $(u_i,v_i)$ be the endpoints of the segment $P_i$. Since $\dist_G(u_i,v_i)\leq 2\beta'$, we have that 
$$\dist_{H'}(u_i,v_i)\leq (1+\epsilon)\dist_{G}(u_i,v_i)+(2k-1)\beta \leq (1+2\epsilon)\dist_{G}(u_i,v_i)~,$$
where the last inequality follows as $\epsilon \cdot \dist_G(u_i,v_i)\geq (2k-1)\beta$. 
Consequently, we have that $\dist_{H'}(u',v')\leq (1+2\epsilon)\dist_{G}(u',v')$. 
\end{proof}

\subsection{Hopsets $\mathbb{\to}$ Weighted Preservers and Spanners}\label{sec:weighted-undirected-hopsets}\vspace{-5pt}
In this section, we show the applications of our approach for weighted undirected graphs and prove \Cref{thm:new-spanners}.
Similarly to the directed setting of Sec. \ref{sec:directed-preservers}, we first translate a hierarchy of (mostly sublinear-sizes) hopsets into missing spanners. The latter can then be converted to near-exact preservers and other variants of spanners. We start by stating the state-of-the-art hopset bounds, show how to translate them (in a black-box manner) into hopsets of \emph{sublinear} size. 

\begin{thm}[Undirected Hopsets]\cite{ElkinN19,HuangP19}\label{thm:undirected-hopset}
For any weighted graph $G=(V,E)$ on $n$ vertices, and any integer $k \geq 1$, there exists an $(\beta,\epsilon)$-hopset $H$ for any $0<\epsilon <1$ with $\beta=O(k/\epsilon)^{k}$ and $|H|=O(n^{1+1/(2^{k+1}-1)})$ edges. Setting $k=\Theta(\log \log n)$, provides linear size hopset with hopbound $\beta=O(\log\log n/\epsilon)^{O(\log\log n)}$. 
\end{thm}
\noindent Henceforth, we assume w.o.l.g that $k \leq \log\log n -1$. In Appendix \ref{sec:missing-proofs}, we show: 
\begin{lem}[Sublinear Undirected Hopsets]\label{lem:undirected-sublinear-hopset}
For any weighted graph $G=(V,E)$ on $n$ vertices, integers $k, D \geq 1$, there exists an $(\beta,\epsilon)$-hopset $H$ for any $0<\epsilon <1$ with $\beta=O(k/\epsilon)^{k}\cdot D$ and $|H|=O((n\log n/D)^{1+1/(2^{k+1}-1)})$ edges.
\end{lem}
\begin{proof}
Let $L$ be a random sample of $O(n\log n/D)$ vertices, obtained by sampling each vertex independently into $L$ with probability of $p=\Theta(\log n/D)$. Let $h(u, v)$ be the minimal number of hops over all $u$-$v$ shortest path in $G$.
Consider a net graph $N=(L,E')$ where 
$$E'=\{(u,v) ~\mid~ u,v \in L, h(u,v)\leq D\}~.$$
Then, let $H$ be an $(\beta',\epsilon)$-hopset for $N$ taken from Theorem \ref{thm:undirected-hopset}. We claim that w.h.p. $H$ is a $(\beta=O(\beta' \cdot D), \epsilon)$-hopset for $G$. 
Consider a $u$-$v$ shortest path $P_{u,v}$ in $G$ of length $10\beta' \cdot D$. Let $u'$ and $v'$ be the closest sampled vertices in $L$ to $u,v$ respectively on the path $P_{u,v}$. 

Then, w.h.p., $h(u,u'),h(v',v)\leq D/2$. 
Let $u'=x_0,x_1, \ldots, x_k=v'$ be the set of sampled vertices in $L$ ordered based on their appearance on $P_{u,v}$. W.h.p., we have that $(x_i,x_{i+1})\in E(N)$ for every $i \in \{0,\ldots, k-1\}$. Therefore, the path $P'=[x_0,x_1, \ldots, x_k] \subseteq E(N)$. Since $H$ is a $(\beta',\epsilon)$-hopset for $N$, it holds that:
$$\dist^{(\beta')}_{N \cup H}(u',v')\leq (1+\epsilon)\len(P')~.$$
As each edge in $N$ corresponds to a $G$-path of $D$ hops, we get that $G \cup H$ contains a $u$-$v$ path of length at most $(1+\epsilon)\len(P)$ with at most $\beta'=O(\beta D)$ hops.
\end{proof}

\noindent \textbf{The Key Step: Hopsets $\to$ Missing Spanners (Weighted, Undirected).} We are now ready to state the main sparsification lemma that computes missing spanners given an undirected hopset hierarchy. In contrast to the directed setting, the size vs. hopbound tradeoff for undirected hopsets is almost tight (up to $n^{o(1)}$ factors). Therefore, we show how translate the \emph{output} hopsets obtained by the state-of-the-art algorithms into missing spanners\footnote{In Sec. \ref{sec:directed-preservers} we considered a general tradeoff function with the purpose that future algorithms that improves the Kogan and Parter bounds \cite{KoganParter22} could immediately derive improved preservers.}.
\begin{lem}[Hopsets $\to$ Missing Spanners]\label{thm:partial-spanner-undirected}
Given $(\beta=O(k/\epsilon)^{k},\epsilon)$ hopsets from Lemma \ref{lem:undirected-sublinear-hopset} and Theorem \ref{thm:undirected-hopset}, one can compute an $\beta$-missing $(1+\epsilon')$-spanner 
$G' \subseteq G$ with $|G'|=\widetilde{O}(n^{1+1/(2^{k+1}-1)} \cdot (c \cdot k/\epsilon)^{2k})$ edges, where $\epsilon' = O(\epsilon \cdot 2^k \cdot (\log \log n -k)) = O(\epsilon \cdot \log n)$ and some constant $c>1$.
\end{lem} 
\begin{proof}
Let $\beta_0=n$ and for every $i \in \{1,\ldots, \ell=\Theta(2^k \cdot (\log \log n -k ))\}$ we set 
\begin{equation}\label{eqn:beta_definition1} \beta_i=\max\left\{\left(\frac{1}{2}\right) \cdot \beta_{i-1}^{\left(1-2^{-k-1}\right)},\Theta((k/\epsilon)^{k})\right\} 
\end{equation}
We will prove now that by the definition of $\beta_i$ we have that  $\beta_\ell = \Theta((k/\epsilon)^{k})$. 
Assume towards a contradiction that $\beta_i =  \left(\frac{1}{2}\right) \cdot \beta_{i-1}^{\left(1-2^{-k-1}\right)}$ for all $1 \leq i \leq \ell$.
Set $c = 1-2^{-k-1}$. Hence as $\beta_0=n$, we have that for $1 \leq i \leq \ell$ we have the following.
\begin{equation}\label{eqn:beta_def}
\beta_i = \left(\frac{1}{2}\right)^{\frac{1-c^i}{1-c}} \cdot n^{c^i} =
 \left(\frac{1}{2}\right)^{2^{k+1} \cdot (1-c^i)}  \cdot n^{c^i}
\end{equation}
Hence, we have 
\begin{align}
\beta_\ell &= \left(\frac{1}{2}\right)^{2^{k+1} \cdot (1-c^\ell)}  \cdot n^{c^\ell} \notag \\ &\leq \left(\frac{1}{2}\right)^{2^{k}} \cdot n^{c^\ell}\leq \left(\frac{1}{2}\right)^{2^{k}} \cdot n^{(\frac{1}{e})^{\Theta(\log \log n - k)}}\leq \left(\frac{1}{2}\right)^{2^{k}} \cdot n^{(\frac{1}{2})^{\log \log n - k}}\leq \left(\frac{1}{2}\right)^{2^{k}} \cdot 2^{2^k} = 1, \notag
\end{align}
where the first inequality follows as $1-c^\ell \geq 1 - \frac{1}{e} > 1/2$. 
This contradicts the definition of Eq. (\ref{eqn:beta_definition1}), concluding that $\beta_\ell = \Theta((k/\epsilon)^{k})$. 

Let $H_i$ be a $(\beta_i,\epsilon)$-hopset for every $i \in \{1,\ldots, \ell\}$, and let $G'$ be the output $r$-missing $t$-spanner obtained by applying Lemma \ref{lem:correspondence} given $\{H_i\}$. We then have that $r=(k/\epsilon)^{k}$ and $t=(1+O(\epsilon \cdot 2^k \cdot (\log \log n -k))$ (by Inequality \ref{ineq:eps}).  
The cardinality of $G'$ is bounded by $|G'|\leq \sum_{i=1}^{\ell} |H_i| \cdot \beta_{i-1}~.$
Using the bounds of Lemma \ref{lem:undirected-sublinear-hopset}
and the fact that for $\beta_i = \omega\left((k/\epsilon)^k\right)$ we have 
$(2\beta_i)^{1+1/(2^{k+1}-1)} = \beta_{i-1}$, we get that:
$$|H_i| \cdot \beta_{i-1}=(n\log n/\beta_i)^{1+1/(2^{k+1}-1)} \cdot  \beta_{i-1} \cdot (k/\epsilon)^{2k}=\widetilde{O}((n\log n)^{1+1/(2^{k+1}-1)} \cdot (k/\epsilon)^{2k}).$$
Therefore, $|G'|=\widetilde{O}(n^{1+1/(2^{k+1}-1)} \cdot (k/\epsilon)^{2k})$. 
\end{proof}

We next provide a collection of applications of Lemma \ref{thm:partial-spanner-undirected} to new constructions of preservers and spanners in weighted undirected graphs.

\paragraph{Weighted Near-Additive Spanners and Near-Exact Preservers.} The next lemma provides a transformation from near-exact hopsets to a variant of near-additive spanners adapted to weighted graphs with maximum edge weight $W_{\max}$. The transformation is based on using missing spanners as an intermediate structure. By using a similar argument to Lemma \ref{lem:missing-spanner-to-pands}, we have the following which proves \Cref{thm:new-spanners}(1):

\begin{lem}[Near-Exact Weighted Hopsets $\to$ Near-Additive Weighted Spanners] \label{lem:hopset-to-weighted-spanner}
Given $(\beta=O(k/\epsilon)^{k},\epsilon)$ hopsets from Lemma \ref{lem:undirected-sublinear-hopset} and Theorem \ref{thm:undirected-hopset} for $k \in \{1,\ldots, \log\log n\}$ and a weighted graph $G=(V,E)$ with maximum edge weight $W_{\max}$, one can compute $(1+\epsilon',\beta')$-spanner $G' \subseteq G$ with $\epsilon'=O( \epsilon \cdot \ 2^k \cdot (\log \log n -k))$, $\beta'=O(\beta \cdot 2^{k} \cdot W_{\max})$ and $|G'|=\widetilde{O}(n^{1+1/(2^{k+1}-1)} \cdot (ck/\epsilon)^{2k})$ edges, for some constant $c>1$. 
\end{lem}
\begin{proof}
Let $G'' \subseteq G$ be the $\beta$-missing $t=(1+\epsilon')$-spanner obtained by applying Theorem \ref{thm:partial-spanner-undirected} given the family of hopsets. The output near-additive spanner $G'$ is given by $G'=G'' \cup G^*$ where $G^* \subseteq G$ is an $(2^{k+2}-3)$-multiplicative spanner for $G$ obtained by Lemma \ref{lem:multiplicative-spanner}. 

Since $|G^*|=O(n^{1+1/(2^{k+1}-1)})$, the size bound holds. We next consider the stretch guarantee. Consider a pair $u,v \in V$.
By the definition of the partial spanner $G''$, the graph $G$ contains a $u$-$v$ path $P_{u,v}$ such that (i) $\len(P_{u,v})\leq (1+\epsilon)\cdot \dist_G(u,v)$ and (ii) $|P_{u,v}\setminus G''|\leq \beta$. Let $E'=P_{u,v}\setminus G''$. For each $e=(x,y) \in E'$, we have that: 
$$\dist_{G'}(x,y) \leq \dist_{G^*}(x,y) \leq O(2^k)\cdot \dist_{G}(x,y) \leq O(2^k \cdot W_{\max})~,$$
where the last inequality follows by the fact that the maximum edge weight in $G$ is at most $W_{\max}$. For unweighted graphs, $W_{\max}=1$.  Since $|E'|\leq \beta$, we get that 
$$\dist_{G''}(u,v)\leq \len(P_{u,v})+\sum_{(x,y) \in E'}\dist_{G'}(x,y) \leq (1+\epsilon')\dist_G(u,v)+O(W_{\max} \cdot \beta \cdot 2^k)~.$$
The stretch argument follows.
\end{proof}
This can be compared with the known bounds for weighted additive weighted spanners, e.g., by Elkin \cite{Elkin01} which obtained a similar tradeoff. In addition, Elkin, Gitlitz and Neiman \cite{EGNarXiv19} provided $(1+\epsilon, \beta W)$ spanners with $\beta=O(k/\epsilon)^{k-1}$ and $O(k n +n^{1/(3/4)^{k}-1})$ edges, of  in which $W$ is the maximum weight on the given $u$-$v$ shortest path.  
%
%
%
%

Turning to near-exact preservers, by using the transformation from missing spanners to preservers of Lemma \ref{lem:missing-spanner-to-pands}, we get the following for weighted and undirected graphs:
\begin{lem}[Near Exact Hopsets $\to$ Near-Exact Preservers]\label{lem:hopset_to_preserver_undirected_weighted}
Given $(\beta=(k/\epsilon)^{k},\epsilon)$ hopsets from Lemma \ref{lem:undirected-sublinear-hopset} and Theorem \ref{thm:undirected-hopset} and a set of $p$ demand pairs $P \subseteq V \times V$, one can compute a $(1+\epsilon')$-approximate preserver $\widehat{G}$ for $P$ with $\epsilon'=O( \epsilon \cdot  2^k \cdot (\log \log n -k))$ and $|\widehat{G}|=\widetilde{O}(n^{1+1/(2^{k+1}-1)} \cdot (ck/\epsilon)^{2k}+ p \cdot \beta)$ edges, for some constant $c>1$.
\end{lem}
\begin{proof}
Let $G'$ be a $\beta$-missing $(1+\epsilon')$-spanner $G' \subseteq G$ obtained by Lemma \ref{thm:partial-spanner-undirected}. For every $(u,v) \in P$, let $P_{u,v} \subseteq G$ be a $u$-$v$ path satisfying that $|P_{u,v}\setminus G'|\leq \beta$ and $\len(P_{u,v})\leq (1+\epsilon')\dist_G(u,v)$. Those paths can be computed in polynomial time based on Lemma \ref{lem:correspondence}. Then, the output preserver $\widehat{G}$ is given by $\widehat{G}= G' \cup \bigcup_{(u,v) \in P} (P_{u,v}\setminus G)~.$

We then have that $P_{u,v} \subseteq \widehat{G}$ for every $(u,v) \in P$, and therefore $\widehat{G}$ is indeed a $(1+\epsilon')$ $P$-preserver. 
The size bound follows as $|\widehat{G}|\leq |G'|+p \cdot \beta$ and by plugging the bound on $G'$ from Lemma \ref{thm:partial-spanner-undirected}. 
\end{proof}
Consequently, by setting $k=\log\log n$, we complete the proof of \Cref{lem:preservers-weighted} which almost matches the bounds obtained for the unweighted case.

%
%

\smallskip

\noindent \textbf{Sourcewise Spanners.} 
Recall that for a source vertex set $S \subseteq V$ in a (possibly weighted) graph $G(V, E)$, an $S$-sourcewise $t$ spanner is a subgraph $H \subseteq G$ such that $\dist_H(u,v) \leq t \cdot \dist_G(u,v)$ for every $u,v \in S \times V$ (see Def. \ref{def:pairwise-spanner}).
 Throughout, we denote $\sigma=|S|$. 
The following Theorem was proven in \cite{ElkinFN17} (formally it is stated for unweighted graphs but the proof holds for weighted graphs without any changes).  
\BTHM\label{thm:arnoldproof}[Theorem 3 of \cite{ElkinFN17}]
For every $k \geq 1$, given an $n$-vertex (possibly weighted) graph $G = (V, E)$ and $S \subseteq V$, there exists a $S$-sourcewise $(4k-1)$-spanner with $O(n + \sqrt{n}\cdot \sigma^{1+1/k})$ edges.
\ETHM
Our improved constructions of sourcewise spanners described next are obtained by simply replacing the exact preservers of \cite{CoppersmithE06} used in the proof of \Cref{thm:arnoldproof} with the near-exact preserver of \Cref{lem:preservers-weighted}. We prove \Cref{thm:new-spanners}(2) by showing:
%
%
\BTHM\label{thm:onlySbound1}
For any parameter $k \geq 1$ and any fixed $\epsilon>0$, for an $n$-vertex (possibly weighted) graph $G=(V,E)$ and $S \subseteq V$ such that $|S|=\sigma$, there exists a $S$-sourcewise $(4k-1+o(1))$-spanner with $(n + \sigma^{1+1/k})\cdot n^{o(1)}$ edges.
\ETHM
\def\APPENDSOURCEWISEFIRST{
\BPF [Proof of \Cref{thm:onlySbound1}]
Consider the complete graph $G'$ on the set of vertices $S$, where each edge $(s,s') \in E(G')$ is weighted by the $s$-$s'$ 
shortest path distance in $G$. Let $H'$ by a $(2k-1)$-spanner $H'$ of $G'$ obtained by using Lemma \ref{lem:multiplicative-spanner}. We then apply the algorithm of \Cref{lem:preservers-weighted} to compute a $(1+\epsilon=o(1))$ distance preserver $G'' \subseteq G$ with the pair set $P=E(H')$, hence $p=|E(H')|$. 
Altogether, we get that $G''$ is a $S$-\emph{subsetwise} $(2k-1)(1+\epsilon)$-spanner for some $\epsilon=o(1)$. I.e., for every $s,s' \in S$ it holds that $\dist_{G''}(s,s'')\leq (2k-1)(1+\epsilon) \dist_{G}(s,s'')$. Finally, the output sourcewise spanner $G^*$ is obtained by computing a shortest path tree $T$ rooted at a dummy vertex $s^*$ in a graph $\widehat{G}=(V \cup \{s^*\}, E \cup \{(s^*,s)~\mid~ s \in S\},W')$ where $W'(e)=W(e)$ for every $e \in E(G)$ and $W'((s^*,s))=1$ for every $s \in S$. 
Define $G^*=G'' \cup (T \cap E(G))$. Adding the edges of $T \cap E(G)$ guarantees that the spanner $G^*$ contains for each vertex $v$ the shortest path from $v$ to its closest source in $S$. 
It is easy to see that $|G^*|\leq |G''|+n=O(n^{o(1)} \cdot (n + \sigma^{1+1/k}))$ by Lemma \ref{lem:hopset_to_preserver_undirected_weighted}. 

We next bound the $s$-$u$ distance in $G^*$ for $s,u \in S \times V$. Let $s_u$ be the closest vertex to $u$ in $S$, then 
\[
\dist_G(s_u, s) \leq \dist_G(s_u, u) + \dist_G(u, s) \leq 2 \dist_G(u, s)~.
\]
We therefore have:
\begin{align}
\dist_{G^*}(u,s) &\leq \dist_{G^*}(u,s_u) + \dist_{G''}(s_u,s) \notag \\ &\leq \dist_G(u,s) + (2k-1)(1+\epsilon) \dist_G(s_u,s) \leq (4k-1)(1+\epsilon)\dist_G(u,s) \notag \\
&\leq (4k-1+\epsilon')\dist_G(u,s) &&\text{for $\epsilon' = (4k-1)\epsilon$},
\notag
\end{align}
where the second inequality follows as $s_u$ is closer to $u$ than $s$, and by the fact that $G''$ is a $S$-subsetwise $(2k-1)(1+\epsilon)$-spanner.
\EPF 
}
We also provide an alternative construction which strictly improves upon \Cref{thm:arnoldproof} of \cite{ElkinFN17} both in terms of the stretch bound and in the number of edges. Note that in comparison to \Cref{thm:onlySbound1}, the obtained stretch is smaller, at the cost of increasing the size.
\BTHM\label{thm:Sandkbound1}
For any parameter $k \geq 2$ and any fixed $\epsilon>0$, for an $n$-vertex (possibly weighted) graph $G=(V,E)$ and $S \subseteq V$ such that $|S|=\sigma$, there exists a $S$-sourcewise $(4k-5+\epsilon)$-spanner with $(n + \sigma\cdot n^{1/k})\cdot n^{o(1)}$ edges.
\ETHM
\BPF 
Set $t=n^\frac{k-1}{k}$ and partition the set $S$ into $q=\lceil \frac{\sigma}{t}\rceil$ sets $S_1,S_2,\ldots,S_q$ such that for all 
$1 \leq i \leq q$ we have $|S_i| \leq t$. 
Now for each $1 \leq i \leq q$, apply Theorem \ref{thm:onlySbound1} with parameter $k-1$ on graph $G$ and the set $S_i$, to obtain a $S_i$-sourcewise $(4k-5+\epsilon)$-spanner with 
the following upper bound on the number of edges 
\[(n + t^{1+1/(k-1)})\cdot n^{o(1)} = n^{1+o(1)}.\]
Hence the total number of edges in the collection of spanners is 
$q \cdot n^{1+o(1)} = (n + \sigma\cdot n^{1/k})\cdot n^{o(1)}$. The stretch argument is immediate. 
\EPF
Finally, we state a conditional lower bound which implies that $S$-sourcewise $(2k-1)$ spanners require $\Omega(|S|\cdot n^{1/k})$ edges. This is in particular interesting in light of the considerably improved size bound obtained in \Cref{thm:onlySbound1}.

\BTHM\label{thm:girthlowerbound1}[Conditional Lower Bounds]
Assuming the Erd\H{o}s girth conjecture, then for every $n,\sigma$ and $k\geq 1$
there is a bipartite graph $G=(V,S,E)$ with $|V|=n$ and $|S|=\sigma$ such that $|E|=\Omega(\sigma\cdot n^{1/k})$ and the removal of any edge $(u,v)$ from $G$ results in a graph $G'$ for which $\dist_{G'}(u,v) \geq 2k+1$.
\ETHM 
\BPF
For each $k\geq1$, by the Erd\H{o}s girth conjecture there is a bipartite graph $G'=(V,U,E')$ such that $|V|=|U|=n$ with $|E'|=\Omega(n^{1+1/k})$ and girth $2k+2$. Create a subgraph $G=(V,S,E)$ of $G'$ by picking each vertex in $U$ into $S$ independently with probability $\frac{\sigma}{n}$. Hence the expected size of $S$ is $\sigma$ and the expected number of edges in $G$ is 
$\Omega\left(\frac{\sigma}{n} \cdot n^{1+1/k}\right) = \Omega(\sigma \cdot n^{1/k})$. The claim follows as the girth of $G$ (and thus also of $G'$) is at least $2k+2$.
\EPF
\noindent Note that the lower bound of Theorem \ref{thm:girthlowerbound1} implies that Theorem \ref{thm:Sandkbound1} is nearly tight for $k=2$.

\smallskip

\noindent \textbf{Spanners with Slack.} The notion of spanners with slack introduced by Chan, Dinitz and Gupta \cite{ChanDG06} provides bounded stretch guarantees for all but a \emph{small} fraction of the vertex pairs. 

\begin{defn}[Spanners with Slack]
Given a (possibly weighted) graph $G$, a subgraph $G' \subseteq G$ is an $\epsilon$-slack $t$-spanner if for each vertex $v \in V(G)$, the farthest $(1-\epsilon)n$ vertices $w$ from $v$ satisfy $\dist_{G'}(v,w)\leq t \cdot \dist_{G}(v,w)$.
\end{defn}

On a high level, the construction of $\epsilon$-slack $t$-spanner presented in \cite{ChanDG06} is based on sampling a small sample of $1/\epsilon$ vertices $N$, such that each vertex is sufficiently close to some sampled vertex. Then, the algorithm computes a multiplicative spanner $H$ on the complete $N \times N$ graph. The final spanner is obtained by (i) connecting each vertex to its nearest vertex in $N$ and (ii) computing an exact preserver for the pairs $P=E(H)$. Specifically, in the lack of near-exact preservers \cite{ChanDG06} used the exact preservers of Coppersmith-Elkin \cite{CoppersmithE06}. Consequently, their work leaves the interesting gap between the size of $\epsilon$-slack emulators vs. spanners. Using the near-exact preservers of \Cref{lem:preservers-weighted}, we almost match this bound. Since the proof is almost identical to that \cite{ChanDG06} and the only difference is in pairwise preserver that we use, we defer the complete analysis to Appendix \ref{sec:missing-proofs} where we show:

\begin{lem}\label{lem:slack-improved}
Given an $n$-vertex (possibly weighted) graph $G=(V,E)$ and parameter $\epsilon\in (0,1)$, there is an $\epsilon$-slack $(5+(6+o(1))\alpha(1/\epsilon))$-spanner with $O(n^{1+o(1)}+T(1/\epsilon)\cdot n^{o(1)})$ edges. 
\end{lem}
\noindent This should be compared with the bound of $O(n+T(1/\epsilon)\cdot \sqrt{n})$ edges shown in \cite{ChanDG06}. 
By plugging Lemma \ref{lem:multiplicative-spanner} into Lemma \ref{lem:slack-improved}, we get the following which proves \Cref{thm:new-spanners}(3):
\begin{cor}\label{cor:slack-final}
For any $\epsilon \in (0,1)$ and integer $k \geq 1$ and $n$-vertex (possibly weighted) graph $G$, there exists an $\epsilon$-slack $(12k-1+o(1))$-spanner $G^* \subseteq G$ with $O(n^{1+o(1)}+(1/\epsilon)^{1+1/k}\cdot n^{o(1)})$ edges.
\end{cor}

\def\APPENDSLACK{
\paragraph{Proof of \Cref{lem:slack-improved}.}
The construction of succinct spanners with slack is based on the notion of density nets. For completeness, we provide the following definitions. 
Let $B(x,r)=\{y ~\mid~ \dist_G(x,y)\leq r\}$ and define $R(x,\epsilon)$ to be the minimum distance $r$ such that $|B(x,r)|\geq \epsilon n$.

\begin{defn}[Density Net]
For $\epsilon \in (0,1)$, an $\epsilon$-net is a set $N \subseteq V$ such that $|N|\leq 1/\epsilon$ and for all $x \in V$ there exists $y \in N$ such that $\dist_G(x,y) \leq 2R(x,\epsilon)$. 
\end{defn}
%
The following lemma, implicit in \cite{ChanDG06}, provides the general recipe for computing $\epsilon$-slack spanners using density nets. 

\begin{lem}[Implicit in \cite{ChanDG06}]\label{lem:slack-implicit}
Given algorithms for computing $\alpha(n)$-spanners with $T(n)$ edges and $\gamma$-multiplicative pairwise spanners for $p$ pairs with $D(p)$ edges, there is an algorithm for computing an $\epsilon$-slack $(5+6\alpha(1/\epsilon)\cdot \gamma)$-spanner $G^*$ with $O(n+D(T(1/\epsilon)))$ edges. 
\end{lem}
\begin{proof}
We provide a proof sketch, see Theorem 5 of \cite{ChanDG06} for additional details. Let $N$ be an $\epsilon$-net for $G$, and let $H'$ be an $\alpha(1/\epsilon)$-spanner for the graph $\widehat{G}=(N, N \times N ,W)$ where $W((z,z'))=\dist_G(z,z')$. Notes that $H'$ is not necessarily a subgraph of $G$. 

Let $G'$ be the set of edges connecting each vertex $v \in V$ to its nearest vertex in $N$. Note that $|G'|=O(n)$. 
Finally, the output spanner $G^*$ is obtained by taking $G'$ and a $\gamma$-multiplicative pairwise spanner $G''$ for the pairs of $E(H')$. Since $p=|H'|=T(1/\epsilon)$, the size bound follows.

For the stretch argument, consider a pair $u,v$ where $v \notin B(u,\epsilon)$ and let $u',v'$ be their nearest vertices in the set $N$. By the definition of $N$, we have that $\dist_G(u,u')\leq 2R(u,\epsilon) \leq 2\dist_G(u,v)$ and $\dist_G(v,v')\leq 3\dist_G(u,v)$. Also, $\dist_{G^*}(u',v')\leq \alpha(1/\epsilon)\cdot \gamma$. Consequently, we have that 
$\dist_{G^*}(u,v)\leq (5+6\alpha(1/\epsilon)\cdot \gamma)\dist_{G}(u,v)$. 
\end{proof}
The construction presented in \cite{ChanDG06} is based on the using the exact preservers of \cite{CoppersmithE06}, leading $\epsilon$-slack $(5+6\alpha(1/\epsilon))$-spanner with $O(n+T(1/\epsilon)\sqrt{n})$ edges. \Cref{lem:slack-improved} follows by plugging the near-exact preservers of \Cref{lem:preservers-weighted} in \Cref{lem:slack-implicit}. 
}

\vspace{-8pt}\section{New Hopset Lower Bounds}\label{sec:LB} \vspace{-8pt}
In this section, we provide a collection of lower bound results on the hopbound of directed and undirected hopsets, and establish Theorem \ref{thm:near-exact-LB}.

\smallskip

\noindent \textbf{Exact Weighted Undirected Hopsets.} We start by considering \emph{exact} hopsets and show how to derive improved hopbound lower bounds by translating lower bounds for exact distance preservers. Specifically, we use the following theorem of \cite{CoppersmithE06}.  

\BTHM\label{thm:coppersmithE1}[Section 3 in \cite{CoppersmithE06}]
For any sufficiently large positive integer numbers $n,p$ where
$\Omega(\sqrt n) = p = O(n^2)$, there exists a weighted undirected graph
$G=(V,E)$ and $P \subseteq  V \times V$ of cardinality $|P|=p$, such that any exact preserver for $P$ requires $\Omega((n\cdot p)^{2/3})$ edges.
\ETHM 

We use the relation between hopsets and preservers shown in Sec. \ref{sec:directed-preservers}, in particular we need the following corollary (restatement of Case 1 in the proof of Theorem \ref{directed_preserver_mainthm1}). 
\begin{cor}\label{directed_preserver_mainthm2}\label{thm:hopsets-to-directed-preservers}
Let $0 \leq b < \min(1,1/(a-1))$ and $a >1$ be fixed parameters. Then, given an algorithm $\mathcal{A}$ for computing $(\beta,\epsilon)$ hopsets for $n$-vertex directed (possibly weighted) graphs with $\widetilde{O}_{W,\epsilon}(n^{2}/\beta^a)$ edges for $\beta = n^b$, the following holds.
For any $n$-vertex directed (possibly weighted) graph $G=(V,E)$ and a set $P$ of $p\leq n^{2-a\cdot b}$ demand pairs, one can compute an $(1+\epsilon')$-distance preserver $G^* \subseteq G$ with 
\[ \widetilde{O}_{W,\epsilon'} (n \cdot p^{1-1/k}) \mbox{  edges where  } k=\frac{2-a\cdot b}{1-b} \mbox{ and } \epsilon'=O(\epsilon \cdot \log \log n)\]  
\end{cor}
\noindent We then show the following for exact hopsets (i.e., $(\beta,0)$ hopsets):
\begin{thm}\label{thm:exact-hopset-wu}[Exact Hopsets, Weighted and Undirected]
For any large enough $n$, there exists an $n$-vertex \textbf{weighted undirected} graph for which any $(\beta,0)$ hopset $H$ with linear size (i.e., $O(n)$) must satisfy that $\beta \geq n^{(1/3 -o(1))}$.  
\end{thm}
\BPF 
Assume towards a contradiction that for any $n$-vertex weighted undirected graph there is a $(\beta,0)$ hopset $H$ such that $|H|=O(n)$ and $\beta = n^{(1/3 - \delta)}$ for some fixed $\delta>0$. Hence by Corollary \ref{directed_preserver_mainthm2} with the following parameters: $b=1/3 - \delta, a=b^{-1}, p=n \mbox{~and~} k = (2-a\cdot b)/(1-b) = 1/(1-b)~,$
we have that for any $n$-vertex undirected weighted graph $G=(V,E)$ and a set $P$ of $n$ demand pairs, one can compute an exact preserver $G^* \subseteq G$ with $\widetilde{O} (n \cdot p^{1-1/k}) = \widetilde{O}(n \cdot n^{b}) = \widetilde{O}(n^{4/3 - \delta}) \mbox{~edges}.$ This contradicts Theorem \ref{thm:coppersmithE1} for $p=n$, as $(n\cdot p)^{2/3} = n^{4/3}$.
\EPF 

\smallskip
\noindent \textbf{Exact Unweighted Undirected Hopsets.} We use the following lower bound result for exact preservers in unweighted undirected graphs by \cite{CoppersmithE06}.

\BTHM\label{thm:coppersmithE2}[The $d=2$ case of Sec. 4 in \cite{CoppersmithE06}]
For any sufficiently large positive integer numbers $n,p$ where
$\Omega(\sqrt n) = p = O(n)$, there exists an unweighted undirected graph
$G=(V,E)$ and a set $P \subseteq V \times V$ of cardinality $|P|=p$,  such that any exact preserver w.r.t. $P$ requires $|E|=\Omega(n^{4/5}\cdot p^{2/5})$ edges.
\ETHM 
\noindent By combining this lower bound with Corollary \ref{directed_preserver_mainthm2}, we get:
\begin{thm}\label{thm:exact-hopset-uu} [Exact Hopsets, Unweighted and Undirected]
For any large enough $n$, there exists an $n$-vertex \textbf{unweighted and undirected} graph for which any $(\beta,0)$ hopset $H$ with linear number of edges must satisfy that $\beta \geq n^{1/5 -o(1)}$.  
\end{thm}
\BPF 
Assume toward a contradiction that for any $n$-vertex unweighted undirected graph there is a $(\beta,0)$ hopset $H$ such that $|H|=O(n)$ and $\beta = n^{(1/5 - \delta)}$ for some fixed $\delta>0$. Hence by Corollary \ref{directed_preserver_mainthm2} with the following parameters $b=1/5 - \delta, a=b^{-1}, p=n \mbox{~and~} k = (2-a\cdot b)/(1-b)= 1/(1-b)$, we have that for any $n$-vertex undirected weighted graph $G=(V,E)$ and a set $P$ of $n$ demand pairs, one can compute an exact preserver $G^* \subseteq G$ with 
$\widetilde{O} (n \cdot p^{1-1/k}) = \widetilde{O}(n \cdot n^{b}) = \widetilde{O}(n^{6/5 - \delta}) \mbox{~edges}.$
This contradicts Theorem \ref{thm:coppersmithE2} for $p=n$, as $n^{4/5} \cdot p^{2/5} = n^{6/5}$.
\EPF 

\smallskip

\noindent\textbf{Near-Exact Hopsets.} We next show that lower bounds for near-exact spanners and emulators can be used to derive lower bounds for hopsets. Abboud, Bodwin and Pettie \cite{AbboudBP18} provided a collection of such lower bound results for both near-exact spanners and hopsets. While the lower bound constructions of \cite{AbboudBP18} obtained for hopsets have a very similar flavor to those provided for spanners and emulators, this is done in a white-box manner. 
The state-of-the-art lower bound results for spanners and emulators are still somewhat stronger for spanners. 

We provide a black-box reduction that converts emulators (or spanners) lower bounds into hopset lower bounds. 
We start by stating the state-of-the-art lower bound results for near-exact hopsets and near-additive spanners by \cite{AbboudBP18}\footnote{Theorem 4.6 of \cite{AbboudBP18} has a minor typo, the fixed statement below has been provided by a personal communication with the authors of \cite{AbboudBP18}.}.

\begin{thm}\label{thm:near-additive-hopsetLB}[Theorem 4.6 of \cite{AbboudBP18}]
Any $(\beta,\epsilon)$ hopset construction with worst-case size at most $n^{1+1/(2^{k}-1)-\delta}$, $\delta>0$, has $\beta=\Omega(1/(\epsilon \cdot exp(k))^{k})$.
\end{thm}

For near-additive emulators, \cite{AbboudBP18} provided stronger lower bound results, e.g., for larger values of $\beta=\Omega(1/(\epsilon k)^{k})$ and for \emph{sublinear} stretch functions. As we will see, our reduction from hopsets to emulators allows us to translate in a black-box manner emulators' lower bounds into hopset lower bounds of nearly the same quality. We use the following lower bound for emulators (which should be compared with Theorem \ref{thm:near-additive-hopsetLB}).

\begin{thm}\label{thm:emulator-LB}[Theorem 2.11, \cite{AbboudBP18}]
The following holds.
\begin{enumerate}
\item 
Any emulator with $f(d)\leq d+O(k \cdot d^{1-1/k}) + \widetilde{O}(1)$ requires $\Omega(n^{1+1/(2^{k}-1)-o(1)})$ edges. 
\item Any $(1+\epsilon,\beta)$ emulator construction with worst-case size at most $n^{1+1/(2^{k}-1)-o(1)}$ has $\beta=\Omega(1/(\epsilon \cdot (k-1))^{k-1})$.
\end{enumerate}
\end{thm}
%
%
\noindent We next derive the following which can be compared with Theorem \ref{thm:near-additive-hopsetLB} of \cite{AbboudBP18}.

\begin{thm}\label{lem:near-exact-LBhopset}[Improved Near-Exact Hopset Lower Bound]
Any $(\beta,\epsilon)$ hopset construction with worst-case size at most $n^{1+1/(2^{k}-1)}/2$ has $\beta \geq (d/\epsilon \cdot (k-1))^{k-1}$ for some sufficiently small constant $d>0$.
\end{thm}
\begin{proof}
Assume towards a contradiction that for every $n$-vertex graph, one can compute a $(\beta,\epsilon)$ hopset $H$ with at most $n^{1+1/(2^{k}-1)-o(1)}/2$ edges and $\beta< (d/\epsilon \cdot (k-1))^{k-1}$. Then, by Lemma \ref{lem:hopset-to-emulator}, one can compute a $(1+\epsilon, \beta')$ emulator $H'$ with
$$|H'|\leq |H|+0.5 \cdot n^{1+1/(2^{k}-1)-o(1)}\leq n^{1+1/(2^{k}-1)} \mbox{~\textbf{and}~} \beta'=2^{k+2} \cdot \beta  \leq (8d/\epsilon \cdot (k-1))^{k-1}~.$$ 

By picking $d$ to be a small enough constant, this leads to a contradiction to Thm. \ref{thm:emulator-LB} (2) (I.e., Theorem 2.11, \cite{AbboudBP18}).
\end{proof}

\begin{thm}\label{thm:lbemulator-to-hopset}
Any hopset with $f(d)\leq d+O(kd^{1-1/k})+\widetilde{O}(1)$ and $\beta =O(kd^{1-1/k})+\widetilde{O}(1)$ requires $\Omega(n^{1+1/(2^{k}-1)-o(1)})$ edges. 
\end{thm}
\begin{proof}
Set $\epsilon = k \cdot d^{-1/k}$. 
Assume towards a contradiction that for every $n$-vertex graph, one can compute a $(1+\epsilon,O(k/\epsilon)^{k-1}+\widetilde{O}(1))$ hopset $H$ with $o(n^{1+1/(2^{k}-1)-o(1)})$ edges. Then, by Lemma \ref{lem:hopset-to-emulator}, one can compute a $(1+\epsilon, \beta')$ emulator $H'$ with
$$|H'|\leq |H|+0.5 \cdot n^{1+1/(2^{k}-1)-o(1)}\leq n^{1+1/(2^{k}-1)-o(1)} \mbox{~\textbf{and}~} \beta'=O(k/\epsilon)^{k-1}+\widetilde{O}(1)~,$$ 
where the bound on $\beta'$ follows as $2^k=O(\log n)$ in our context. 
By picking the constants hidden in the $O$-notation to be a small enough, this leads to a contradiction to Thm. \ref{thm:emulator-LB} (1) (I.e., Theorem 2.11, \cite{AbboudBP18}).
\end{proof}

Theorem \ref{thm:near-exact-LB} follows by Theorems \ref{thm:exact-hopset-wu},\ref{thm:exact-hopset-uu},\ref{thm:lbemulator-to-hopset} and \ref{lem:near-exact-LBhopset}. Finally, we show that by using the general reduction of Theorem \ref{directed_preserver_mainthm1}, we can also recover known lower bounds for reachability preservers. 

\smallskip

\noindent \textbf{Lower Bounds for Reachability Preservers Imply Lower Bounds for Directed Shortcuts.}  
Finally, we show how to translate lower bounds for reachability preservers into lower bounds on the diameter of linear size shortcuts.

\BPF [Proof of Theorem \ref{thm:general_shortcut_bound1}]
Assume by contradiction that for any $n$-vertex directed graph there is a $d$-shortcut $H$ such that $|H|=O(n)$ and $d = n^{(\alpha+\gamma -1 - \delta)}$ for some fixed $\delta>0$. Hence by Corollary \ref{directed_preserver_mainthm2} with the following parameters:
$$ b= \alpha+\gamma -1 - \delta,~~ a=b^{-1}, ~~p=n \mbox{~and~} k = \frac{2-a\cdot b}{1-b} = \frac{1}{1-b}~.$$

We have that for any $n$-vertex directed graph $G=(V,E)$ and a set $P$ of $n$ demand pairs, one can compute a reachability preserver $G^* \subseteq G$ with 
\[ \widetilde{O} (n \cdot p^{1-1/k}) = \widetilde{O}(n \cdot n^{b}) = \widetilde{O}(n^{\alpha + \gamma - \delta}) \]  edges. This contradicts  the assumed lower bound of $\Omega(n^{\alpha} \cdot p^{\gamma})$ edges for reachability preservers (for $p=n$) and thus we are done.
\EPF 

The following lower bound for reachability preservers have been shown by Abboud and Bodwin \cite{AbboudB18}. 
\BTHM\label{thm:reachability_lowerbound1}[Theorem 2.6 in \cite{AbboudB18}]
There  is  an infinite  family  of $n$-node  directed graphs  and  pair  sets $P$ for which every reachability preserver of $G$ with respect to the set $P$ requires $|E|=\Omega(n^{2/3} \cdot p^{1/2})$ edges.
\ETHM 

By plugging Theorem \ref{thm:reachability_lowerbound1} in Theorem \ref{thm:general_shortcut_bound1} (i.e., taking $\alpha=2/3,\gamma=1/2$), we recover the state-of-the-art lower bound by Huang and Pettie \cite{HuangP18}:

\BTHM\label{thm:shortcut_lowerbound1}[Theorem 2.1 in \cite{HuangP18}]
For any large enough $n$, there exists an $n$-vertex directed graph for which any $d$-shortcut $H$ with linear number of edges must satisfy that $d=n^{(1/6 -o(1))}$.  
\ETHM

\paragraph{Acknowledgment.} We are grateful to Amir Abboud, Aaron Bernstein, Greg Bodwin and Nicole Wein for useful discussions. 

\bibliographystyle{alpha}
\bibliography{thesis}

\newpage

\begin{appendix}

\section{Missing Proofs}\label{sec:missing-proofs}






\APPENDSOURCEWISEFIRST



\APPENDSLACK

\end{appendix}

\end{document}